\documentclass[onecolumn,11pt,draftcls]{IEEEtran}
\usepackage{amsbsy}
\usepackage{floatflt} 

\usepackage{amsmath}
\usepackage{amssymb}
\usepackage{times}
\usepackage{graphicx}
\usepackage{xspace}
\usepackage{paralist} 
\usepackage{setspace} 
\usepackage{xypic}
\xyoption{curve}
\usepackage{latexsym}
\usepackage{theorem}
\usepackage{ifthen}
\usepackage{subfigure}
\usepackage{booktabs}

\newcommand{\mypar}[1]{\vspace{0.03in}\noindent{\bf #1.}}

\newcounter{brojac}

{\theoremheaderfont{\it} \theorembodyfont{\rmfamily}
\newtheorem{assumption}[brojac]{Assumption}
}
{\theoremheaderfont{\it} \theorembodyfont{\rmfamily}
\newtheorem{theorem}{Theorem}
\newtheorem{lemma}[theorem]{Lemma}
\newtheorem{definition}[theorem]{Definition}

\newtheorem{fact}[theorem]{Fact}

}

\begin{document}
\title{Distributed Detection over Noisy Networks: Large
Deviations Analysis}

\author{Du$\breve{\mbox{s}}$an Jakoveti\'c, Jos\'e M.~F.~Moura, Jo\~ao
Xavier
\thanks{
The work of the first, and the third authors is partially supported by: the Carnegie Mellon|Portugal
Program under a grant from the Funda\c{c}\~{a}o para a Ci\^encia e Tecnologia (FCT) from Portugal; by FCT grants CMU-PT/SIA/0026/2009; and by ISR/IST plurianual funding (POSC program,
FEDER). The work of J.~M.~F.~Moura is partially supported by NSF under grants CCF-1011903 and CCF-1018509, and by AFOSR grant
FA95501010291. D. Jakoveti\'c holds a fellowship from FCT.}
\thanks{D. Jakoveti\'c is with the
Institute for Systems and Robotics
(ISR), Instituto Superior T\'{e}cnico (IST), Technical University of Lisbon, Lisbon, Portugal, and with
the Department of Electrical and Computer Engineering, Carnegie Mellon
University, Pittsburgh, PA, USA {\tt\small
djakovet@andrew.cmu.edu}}%
\thanks{J. M.~F.~Moura is with the Department of
Electrical and Computer
Engineering, Carnegie Mellon University, Pittsburgh, PA, USA {\tt\small
moura@ece.cmu.edu}}%
\thanks{J. Xavier is with the Institute for Systems and Robotics (ISR),
Instituto Superior T\'{e}cnico (IST), Technical University of Lisbon, Lisbon, Portugal {\tt\small
jxavier@isr.ist.utl.pt}}}


\maketitle

\begin{abstract}
We study the large deviations performance of
consensus+innovations distributed detection over \emph{noisy} networks, where sensors at a time step $k$
cooperate with immediate neighbors (consensus)
  and assimilate their new observations (innovation.) We
  show that, even under noisy communication, \emph{all sensors} can
  achieve exponential decay $e^{-k C_{\mathrm{dis}}}$ of the
  detection error probability,
   even when certain (or most)
    sensors cannot detect the event of interest
    in isolation. We achieve this by
    designing a single time scale stochastic
    approximation type distributed detector with the optimal weight sequence $\{\alpha_k\}$,
      by which sensors weigh their neighbors' messages.
      The optimal design of $\{\alpha_k\}$ balances the opposing effects of
       communication noise and information flow from neighbors:
        larger, slowly decaying $\alpha_k$ improves information flow
         but injects more communication noise.
      Further, we quantify the best
      achievable $C_{\mathrm{dis}}$
       as a function
       of the sensing signal and noise,
       communication noise, and network connectivity.
       Finally, we find a threshold
       on the communication noise power
       below which a sensor that can detect the event
       in isolation still improves its detection by cooperation through noisy links.

\end{abstract}
\hspace{.43cm}\textbf{Keywords:} Noisy communication, distributed detection, Chernoff information, large deviations.
\newpage
\section{Introduction}
\label{section-intro}
Consider a generic network of sensors that sense the environment
 to detect an event of interest.
  In centralized detection, the measurements of \emph{all} sensors at \emph{all} times $k$ are available at the
  detector (fusion center.) Under appropriate conditions, the probability of error $P^e(k)$ of the centralized minimum probability of error detector decays in $k$ at an exponential rate, $P^e(k)\sim e^{-C \,k}$, where~$C$ is the (centralized)
  Chernoff information. We research in this paper the equivalent question  of exponential rate of decay of the probability of error~$P^e$ for \emph{distributed} detection at each local sensor. We consider this question when the (local) communications among sensors is through \emph{noisy} links.

  To be specific, we study consensus+innovations distributed algorithms like for example the LMS and RLS adaptive algorithms in~\cite{cattivellisayed-2011,Sayed-detection,Sayed-detection-2},
  the detectors in~\cite{running-consensus-detection,Soummya-Detection-Noise}, or the estimators in~\cite{SoummyaEst}. In consensus+innovations distributed algorithms, at time~$k$, each
 sensor updates its state
 \begin{inparaenum}[1)]
 \item by a weighted average of the states of its neighbors (consensus); and
 \item by incorporating its local measurement (innovations):
 \end{inparaenum}
  \begin{equation}
  \label{eqn:consensus+innovations}
  \mathrm{state}_{\,k+1}=\mathrm{state}_{\,k} + \gamma_k^1\,\, \mathrm{consensus}_{\,k} + \gamma^2_k\,\,
  \mathrm{innovations}_{\,k}.
  \end{equation}
  Consensus+innovations detectors like in~\eqref{eqn:consensus+innovations} are distributed,
 stochastic approximation type algorithms, but particular algorithms make different choices of the \emph{time-decaying weight sequences} $\gamma_k^i$, $i=1,2$, in~\eqref{eqn:consensus+innovations} by which sensors weigh the consensus (their neighbors' messages) and the innovations (their own measurements) terms at each time~$k$: \cite{cattivellisayed-2011,Sayed-detection,Sayed-detection-2} set $\gamma_k^i=\mu$, $i=1,2$; in~\cite{running-consensus-detection} they vanish at the same rate; while~\cite{Soummya-Detection-Noise,SoummyaEst} consider single but also mixed scale algorithms where these weight sequences vanish at different rates.   We will show that key to achieving exponential decay of the distributed detector $P^e$ at \emph{all} sensors is the suitable design of the weights $\gamma_k^i$, $i=1,2$, in~\eqref{eqn:consensus+innovations}.


 This paper addresses three natural questions:
%
%
   %
   %
   \begin{enumerate}
  \item \emph{Centralized versus distributed (through noisy links) detection}: under which, if any,
  conditions can consensus+innovations distributed detection achieve \emph{exponentially fast} decay of the (detection) error probability \emph{at each sensor}, $P^e\sim e^{-C_{\mathrm{dis}} k}$--the best possible decay rate~$C_{\mathrm{dis}}$, not necessarily equal to the centralized Chernoff information $C$; in other words, when the sensors cooperate through noisy links, can distributed detection achieve at every sensor an exponential rate of decay like centralized detection. We answer affirmatively this question, under mild structural conditions, by careful design of the weight sequences in~\eqref{eqn:consensus+innovations}.
  \item \emph{Can cooperation through noisy links help}: How close can $C_{\mathrm{dis}}$ approach the corresponding (centralized) Chernoff information $C$? We solve this by designing the \emph{best} weight sequences $\gamma_k^i$, $i=1,2$, that maximize $C_{\mathrm{dis}}$. We explicitly quantify the optimal $C^\star_{\mathrm{dis}}$ as a function of a \emph{sensing} signal-to-noise ratio~\textbf{SSNR} and a \emph{communication} signal-to-noise ratio~\textbf{CSNR} that we will define, and the network algebraic connectivity\footnote{The algebraic connectivity is the second smallest eigenvalue of the graph Laplacian matrix~\cite{chung} that measures the speed of averaging across the network; larger algebraic connectivity means faster averaging.}. Our analysis reveals opposing effects: small and fast decaying weight $\gamma_k^1$ injects less communication noise, but also reduces the information flow from neighbors; the optimal weights strike the best balance between these two effects.
  \item \emph{\textbf{SSNR} versus \textbf{CSNR}--how much communications noise can distributed detection sustain}: What is the highest communications noise level, i.e., lowest \textbf{CSNR}, for which cooperation helps? Let sensor~$i$ be the \emph{best} sensor among all locally detectable sensors and assume that, without cooperation, its $P^e\sim e^{-k \,C_i}$. Can cooperation over noisy links make the worst sensor under communication better than $i$--the best one without cooperation? We explicitly find a \emph{threshold} on the ratio \textbf{CSNR}/\textbf{SSNR} above which communication pays off in the latter sense; the threshold is a function of the network algebraic connectivity.
\end{enumerate}

\mypar{Brief comment on the literature} \emph{Consensus}+\emph{innovations} distributed algorithms as in~\eqref{eqn:consensus+innovations} or in~\cite{running-consensus-detection,cattivellisayed-2011,Sayed-detection,Sayed-detection-2,Soummya-Detection-Noise,SoummyaEst,Stankovic},
 that interleave \emph{consensus} and \emph{innovations} at the same time instant contrast with decentralized parallel fusion architectures,  e.g., \cite{Varshney-I,Veraavali,Tsitsiklis-detection,Poor-II,Moura-saddle-point}, where all sensors communicate with a fusion sensor or with consensus-based detection schemes (no fusion sensor,) for example, \cite{Moura-detection-consensus,moura-cons-detection,consensus-detection,weight-opt}, where sensors in the network, initially,
 \begin{inparaenum}[1)]
 \item collect a single snapshot of measurements, and, \emph{subsequently},
 \item run the consensus algorithm to fuse their decision rules.
 \end{inparaenum}

 We consider noisy communications among sensors. Communications imperfections in consensus-based detection in sensor networks are usually modeled via intermittent link failures and additive noise~\cite{SoummyaChannelNoise}. For consensus+innovations distributed algorithms, the LMS and RLS adaptive algorithms in~\cite{cattivellisayed-2011,Sayed-detection,Sayed-detection-2} and
 the distributed change detection algorithm in \cite{Stankovic} do not consider link failures nor additive noise. References \cite{running-consensus-detection,dis-det-ICASSP,dis-det-TSP} consider link failures but no additive communication noise. Reference \cite{dis-det-Allerton} considers deterministically time varying
 networks. Reference~\cite{SoummyaEst} is concerned with estimation and considers a very general model that includes sensor failures, link failures, and various degrees of either quantized or noisy communications.  To the best of our knowledge and within the consensus+innovations detectors, only reference \cite{Soummya-Detection-Noise} and now this paper consider additive noise in the communications among sensors, but no link failures. We highlight the main differences between our work and~\cite{Soummya-Detection-Noise}.

 Reference \cite{Soummya-Detection-Noise} proposes a consensus+innovations distributed detector that it refers to as $\mathcal{MD}$. Algorithm $\mathcal{MD}$ assumes very general data distributions: temporally independent, spatially correlated sensing noise and temporally independent, spatially correlated additive communication noise, both with generic distributions with finite second moments. Under global detectability and connectedness assumptions, \cite{Soummya-Detection-Noise} shows that $\mathcal{MD}$'s error probability $P^e$ decays to zero at all nodes~$i$, but \cite{Soummya-Detection-Noise} only shows exponential decay rate of the error probability for a modified, $\mathcal{SD}$ scheme, when the noises are Gaussian, and \emph{all sensors are locally detectable, with equal Chernoff informations $C_i>0$}\footnote{Sensor $i$ can detect the event individually (is locally detectable) if and only if $C_i>0$; see ahead Definition~1 and Fact~2.} (\cite{Soummya-Detection-Noise}, Corollary 12.) In fact, as we show in this paper, Appendix \ref{subsect-soummya}, the probability of error~$P^e$ for $\mathcal{MD}$ is not exponential; it is instead \emph{sub exponential}\footnote{We show in Appendix \ref{subsect-soummya} that $\max_{i=1,...,N} P^e_i(k)$--the worst error error probability at time~$k$ among all $N$ sensors, is at least $e^{-c k^\tau}$, where $\tau \in (0.5,1)$ and $c>0$.}, i.e., the rate is strictly slower than exponential, when the $C_i$'s are not equal, with possibly some $C_i$'s equal to zero. The subexponential rate of the $\mathcal{MD}$ and $\mathcal{SD}$ algorithms (with unequal $C_i$'s) is due to the decay rates assumed by these algorithms for the stochastic approximation weight sequences. In contrast, in the consensus+innovations algorithm that we propose, we craft carefully these weight sequences; this enables us to show for the Gaussian problem and under global detectability and connectedness that our distributed detector achieves exponential decay rate for $P^e$ at every sensor, regardless of the equal or unequal $C_i$'s, where some can possibly be zero. Further, we optimize the weight sequences so that sensors achieve the maximum payoff from their (noisy) cooperation with other sensors. We derive our results on the $P^e$ under Gaussian assumptions on the sensing and communication noises, but our results extend, to a certain degree, to the \emph{non-Gaussian} (time-independent and space-independent) zero mean sensing noise with finite second moment and to the \emph{non-Gaussian} (time-independent and space-correlated) zero mean communication noise with finite second moment.

   The Gaussian assumptions allow us to completely characterize the rate of decay of the error probability solely on the basis of \emph{the first two moments} (mean and variance) of the node's decision variable or state, say $x_i(k)$. With non-Gaussian noises, the first two moments no longer suffice to determine the rate of decay of the error probability $P^e$, but they still represent a good measure of detection performance. In this case, we can still show that the (local) detector signal-to noise ratio ${\textbf{DSNR}}_i(k)$ at each sensor $i$ that we define by the ratio of the square of the mean over the variance of the sensor~$i$ state~$x_i(k)$ grows at the same rate $\sim k$, as for the optimal centralized detector.


We relate now this paper to our prior work on distributed detection, \cite{dis-det-ICASSP,dis-det-TSP}.
While~\cite{dis-det-ICASSP,dis-det-TSP} study the effect of link failures on detection performance, this paper addresses \emph{additive communication noise} in the links among sensors. Our analysis here reveals that communication noise has an effect that is qualitatively different from that of link failures; with link failures, the more communication that is actually achieved among sensors the better the error performance, since when communication does happen sensors receive their neighbors decision variables unencumbered by noise. On the other hand, communication noise leads to a clear tradeoff between communication noise and information flow (degree at which consensus helps), with a cooperation payoff--threshold on the \textbf{CSNR}. To show these results, the analysis we develop here is very different from the analysis we advanced in \cite{dis-det-ICASSP,dis-det-TSP}. In~\cite{dis-det-ICASSP,dis-det-TSP}, we consider independent identically distributed (i.i.d.) averaging matrices $W(k)$ (and hence the distribution of the $W(k)$ is time invariant,) and no communication noise. In contrast, this paper considers time-decaying stochastic approximation weights (and hence, time varying weight matrices $W(k)$) and additive communication noise; these additional challenges do not allow for our tools in~\cite{dis-det-ICASSP,dis-det-TSP} and demand new analysis. A final comment to distinguish our methods here with those in~\cite{Soummya-Detection-Noise}. Reference~\cite{Soummya-Detection-Noise} uses standard stochastic approximation techniques~\cite{stochastic-approx} that yield the exact asymptotic covariance of the decision variable vector (when the local Chernoff informations are equal) the asymptotic covariance is given by a difficult to interpret matrix integration formula. In contrast, we do not pursue the exact asymptotic covariance of the decision variable, but get, instead, tight, simple, easy to interpret lower and upper bounds by exploiting the natural separability between the communication noise and the information flow (averaging) effects.

\mypar{Paper organization} The next paragraph introduces notation. Section \ref{section-prob-model}
describes the problem model and presents our distributed detector.
 Section \ref{section-modeling} states our modeling assumptions and gives
 preliminary analysis. 
 Section \ref{section_perf_analysis} presents our main
 results on the asymptotic performance of our distributed detector. 
 Section \ref{section-proofs} proves our main results.
  Section \ref{section-extensions} presents extensions to the non-Gaussian case.
  Finally, section \ref{section-concl} concludes the paper.
  Appendices A--B provide remaining proofs.

\mypar{Notation} We denote by: $A_{ij}$ or $\left[ A\right]_{ij}$ (as appropriate) the $(i,j)$-th entry of a matrix $A$; $a_i$
 or $[a]_i$ the $i$-th entry of a vector $a$; $A^\top$ the transpose of $A$; $I$, $1$, and $e_i$, respectively, the identity matrix, the column vector with unit entries, and the
$i$-th column of $I$,  $J:=(1/N)11^\top$ the $N \times N$ ideal averaging matrix; $\| \cdot \|_l$
the vector (respectively, matrix) $l$-norm; $\|\cdot\|=\|\cdot\|_2$ the Euclidean (respectively, spectral) norm;
 $\lambda_i(\cdot)$ and $\mathrm{tr}(\cdot)$ the $i$-th smallest eigenvalue, and the trace of a matrix;
  $\mathrm{Diag}\left(a\right)$ the diagonal matrix with the diagonal equal to the vector $a$;
 $|\mathcal{A}|$ the cardinality of $\mathcal{A}$; $\mathbb E \left[ \cdot \right]$, $\mathrm{Var}(\cdot)$, $\mathrm{Cov}(\cdot)$,
and $\mathbb P \left( \cdot\right)$ the expected value, the variance, the covariance, and probability operators,
 respectively; $\mathcal{I}_{\mathcal{A}}$ the indicator function of the event $\mathcal{A}$; $\mathcal{N} \left( \mu, \Sigma\right)$ the normal distribution
with mean $\mu$ and covariance $\Sigma$; $\mathcal{Q}(\cdot)$ the Q-function, i.e., the function that calculates the right tail probability
 of the standard normal distribution;
 \begin{equation}
 \label{eqn:qfunction}
 \mathcal{Q}(t)=\frac{1}{\sqrt{2 \pi}} \int_{t}^{+\infty} e^{-\frac{u^2}{2}}du,\:\:t \in \mathbb R.
 \end{equation}
 We also make use of the standard~$\Omega$ and~$O$ notations: $f(k) = \Omega\left( g(k)\right)$ stands for existence of a $K>0$ such that $f(k) \geq c g(k)$, for some $c>0$, for all $k \geq K$; and
  $f(k) = O\left( g(k)\right)$ means existence of $K>0$ such that $f(k) \leq c g(k)$, for some $c>0$, for all $k \geq K$.
\section{Binary Hypotheses Testing: Centralized and Distributed}
\label{section-prob-model}
This section presents the network model
 and our consensus+innovations distributed detector whose performance we analyze in section \ref{section_perf_analysis}.
  The current section also considers the centralized and isolated sensor
  detectors for benchmarking our consensus+innovations detector
   and defines certain relevant signal-to-noise ratios.

\vspace{-3mm}
\subsection{Network}
\label{subsect-netw-data-mixing-model}
%
We consider a network of~$N$ sensors. The topology of the network defines who can communicate with
whom and is described by a simple (no self or multiple links,) undirected graph $\mathcal{G}=\left(\mathcal{V},\mathcal{E} \right)$, where
$\mathcal{V}$ is the set of sensors with $|\mathcal{V}|=N$, and $\mathcal{E}$ is the set of links or communication channels among sensors: the link between sensors~$i$ and $j$ is represented in the graph by $(i,j)\in\mathcal{E}$. 
%
%
  %
  %
The neighborhood set $O_i$ and the degree $d_i$ of sensor $i$ are $O_i =\{j: (i,j) \in E\}$ and $d_i=|O_i|$, respectively. The $N \times N$ adjacency matrix is $A=\left[A_{ij}\right]$, with $A_{ij}=1$ if $(i,j)\in E$
 and $A_{ij}=0$ else (with $A_{ii}=0$, for all $i$.) The graph Laplacian is $\mathcal{L}=D-A$, where $D=\mathrm{Diag} \left(d_1,...,d_N \right)$.
 Consider the eigenvalue decomposition
\begin{equation}
\label{eqn:Laplaciandecomposition}
\mathcal{L} = Q \, \Lambda(\mathcal{L}) \, Q^\top,
\end{equation}
where $\Lambda(\mathcal{L}) = \mathrm{Diag} \left( \lambda_1(\mathcal{L}),...,\lambda_N(\mathcal{L})\right)$, with the
eigenvalues in increasing order, and the
 columns $q_i$ of~$Q$ are the orthonormal eigenvectors of~$\mathcal{L}$.
It is well known that $\lambda_1(\mathcal{L})=0$ and
$
q_1 = \frac{1}{\sqrt{N}}1.$ Further, $\mathcal{G}$ is connected if and only if the algebraic connectivity $\lambda_2(\mathcal{L})>0$,~\cite{chung}.
  %
  %
%
  %
  %
\vspace{-3mm}
\subsection{Isolated sensor detector and centralized detector}
We consider the known signal in Gaussian noise binary hypotheses test between $H_1$ and $H_0$. At time $k$, sensor $i$  measures the (scalar) $y_i(k)$:
%
%
\label{assumption-sen-obs}
\begin{equation*}
\mathrm{under}\,\, H_l:\,\,y_i(k) = [m_l]_i + \zeta_i(k),\,\,l=0,1,
\end{equation*}
with prior probabilities $0<P(H_1)=1-P(H_0)<1$. Here $[m_l]_i$ is a constant known signal and the \emph{sensing} noise $\{\zeta_i(k)\}$ is a zero mean (z.m.) independent identically distributed (i.i.d.) Gaussian sequence.
%
Introduce the vector notation
 \begin{eqnarray*}
 \label{eqn-y(k)}
 y(k) &=&\left(y_1(k),\cdots,y_N(k)\right)^\top,\:
 m_l=\left(\left[m_l\right]_1,\cdots,\left[m_l\right]_N\right)^\top,\:
 \zeta(k) =\left(\zeta_1(k),\cdots,\zeta_N(k)\right)^\top.
 \end{eqnarray*}
 The covariance of the \emph{sensing} noise is $  S_{\zeta}=\mathrm{Cov}\left(\zeta(k)\right).$
%
%
%

\mypar{Isolated sensor detector}
A sensor working in isolation processes only its own observation. The test statistic is the local likelihood ratio ($\ell$LLR) $\mathcal{D}_i(k)$, $k=1,\cdots,$ given by the sum of the instantaneous $\ell$LLR $L_i(j)$, $j=1,\cdots,k$, where
\begin{eqnarray}
\label{eqn:isolateddetector-1}
L_i(j)&=&[m_1-m_0]_i\, \frac{ y_i(j) -  \frac{\left[ m_1 \right]_i+\left[ m_0 \right]_i}{2} } {[S_{\zeta}]_{ii}} \nonumber \\
\label{eqn:isolateddetector-2}
\mathcal{D}_i(k)&=&\frac{1}{k}\sum_{j=1}^k L_i(j).
\end{eqnarray}
The isolated sensor~$i$ detector thresholds $\mathcal{D}_i(k)$ against a threshold $\tau_i(k)$.

\mypar{Centralized detector}
The centralized log-likelihood ratio~(\textit{c}LLR) for the single vector measurement $y(k)$ (all sensors measurements are available at the fusion center) is:
\begin{eqnarray*}
\label{eqn_LLR}
L(k) = (m_1-m_0)^\top {S_{\zeta}}^{-1} \left(y(k) - \frac{m_1+m_0}{2}\right).
\end{eqnarray*}
 The \textit{c}LLR at time~$k$ is:
\begin{equation}
\label{eqn-d-k}
\mathcal{D}(k) = \frac{1}{k} \sum_{j=1}^k L(j)
\end{equation}
The optimal centralized detector thresholds the \textit{c}LLR against $\tau(k)$.
For future reference we introduce:
\begin{eqnarray}
\label{eqn:eta}
\eta(k) &=& (\eta_1(k),\,\eta_2(k),\,...\,,\eta_N(k))^\top\\
\label{eqn-eta-i}
\eta_i(k) &=& \left[ {S_{\zeta}}^{-1}(m_1-m_0)\right]_i \left( y_i(k) - \frac{[m_1]_i+[m_0]_i}{2} \right).
\end{eqnarray}
Conditioned on $H_l$, $l=0,1$, the sequence $\eta(k)$ is i.i.d.~Gaussian with mean $m_{\eta}^{(l)}$ and
covariance $S_{\eta}$:
\begin{eqnarray}
\label{eqn_mu_sigma}
m_{\eta}^{(l)} &=& (-1)^{(l+1)} \mathrm{Diag} \left( S_{\zeta}^{-1}(m_1-m_0)
\right)\,\frac{1}{2}(m_1-m_0)\\
\label{eqn_eta_sigma}
S_{\eta} &=& \mathrm{Diag} \left( S_{\zeta}^{-1}(m_1-m_0) \right) \, S_{\zeta}  \,\mathrm{Diag} \left( S_{\zeta}^{-1}(m_1-m_0) \right).
\end{eqnarray}

With~\eqref{eqn-eta-i}, we rewrite the \textit{c}LLR $L(k)$ at time~$k$ as the separable sum of $\eta_i(k)$'s:
\begin{equation*}
\label{eqn-L(k)-summable}
L(k) = \sum_{i=1}^N \eta_i(k).
\end{equation*}
\vspace{-5mm}
\subsection{Distributed detector: Consensus+Innovations}
\label{subsec:distributeddetector}
We now consider the consensus+innovations distributed detector, see~\eqref{eqn:consensus+innovations}, with structure like the structure of the distributed estimator in~\cite{SoummyaEst} or of the $\mathcal{MD}$ distributed detector in \cite{Soummya-Detection-Noise}. The key to ours is our choice of the consensus weight sequence $\gamma_k^1$ that we show to be the optimal one and will lead to exponential decay rate of the probability of error of the consensus+innovations distributed detector, which is not the case in general for the $\mathcal{MD}$ distributed detector in \cite{Soummya-Detection-Noise} as we show in Appendix~\ref{subsect-soummya}.

To set-up the distributed detector, let the decision variable or the current state of sensor~$i$ at time~$k$ be $x_i(k)$. Due to the \emph{communication} noise, when sensor~$j$ transmits to sensor~$i$ its state, sensor~$i$ receives a noisy version:
\vspace{-2mm}
    \begin{equation}
    \label{eqn:commnoise-0}
    x_j(k)+\nu_{ij}(k),
    \end{equation}
where $\nu_{ij}(k)$ is the \emph{communications} noise. Note that \eqref{eqn:commnoise-0} is a high level model, i.e.,
 we do not model here the physical communication channel, but rather we model the estimation errors at the receiver.

We propose as distributed detector the single time scale, stochastic approximation, consensus+innovations algorithm where each sensor updates its decision variable two-fold:
\begin{inparaenum}[1)]
\item by \emph{consensus}, i.e., averaging its decision variable with
 the decision variables of its immediate neighbors---the sensors with which it communicates; and
 \item by \emph{innovation}, i.e., by incorporating the innovation $\eta_i(k)$
 in~\eqref{eqn-eta-i}, after \emph{sensing} its local observation.
 \end{inparaenum}
 The consensus+innovation update of $x_i(k)$ is given by:
\begin{eqnarray}
\label{eqn-running-cons-sensor-ii}
x_i(k+1) \hspace{-1.5mm} &=& \hspace{-.5mm} \hspace{-2mm} x_i(k)+ \frac{k}{k+1}
\alpha_k \underbrace{\sum_{j \in O_i} \left(\left(x_j(k)-x_i(k)\right) + \nu_{ij}(k)\right)}_{\mathrm{consensus}} +
\frac{1}{k+1} \underbrace{\left(\eta_i(k+1)-x_i(k)\right)}_{\mathrm{innovations}} \\
\label{eqn-running-cons-sensor-i}
 &=& \hspace{-.5mm} \hspace{-2mm} \frac{k}{k+1}  \left(1 - \alpha_k d_i  \right) x_i(k)+ \frac{k}{k+1}
\alpha_k \sum_{j \in O_i} \left( x_j(k) + \nu_{ij}(k) \right) + \frac{1}{k+1} \eta_i(k+1) \\
\label{eqn-running-cons-sensor-iii}
 x_i(1) &=& \eta_i(1). \nonumber
\end{eqnarray}

We write \eqref{eqn-running-cons-sensor-i} in matrix form. The \emph{communication} noise
    at sensor~$i$ from all its neighbors, and the corresponding vector quantity for all sensors, are (see~\eqref{eqn-running-cons-sensor-i})
    \begin{equation}
    \label{eqn:commnoise-1}
    v_i(k):=\sum_{j \in O_i} \nu_{ij}(k),\:\:
    v(k):=\left(v_1(k),\cdots,v_N(k)\right)^\top.
    \end{equation}

    Denote also by $    S_v:=\mathrm{Cov} \left( v(k) \right). $
%
 Likewise, define the vector of the sensors decision variables or vector of sensor states
   $
    x(k):=\left(x_1(k),\cdots, x_N(k)\right)^\top
   $
 and the time varying, deterministic, averaging matrix
$
W(k):=I-\alpha_k \mathcal{L},
$
where $\mathcal{L}$ is the graph Laplacian, see Subsection~\ref{subsect-netw-data-mixing-model}. Then \eqref{eqn-running-cons-sensor-i}~is:
\begin{eqnarray}
\label{eqn_recursive_algorithm}
x(k+1) = \frac{k}{k+1} W(k) x(k) + \frac{k}{k+1} \alpha_k v(k) + \frac{1}{k+1} \eta(k+1), \:\: k=1,2,...,\:
 x(1)  = \eta(1),
\end{eqnarray}
where $\eta(k)$ is given in~\eqref{eqn:eta}. When the noises are Gaussian, since~\eqref{eqn-running-cons-sensor-i} and~\eqref{eqn_recursive_algorithm} are linear, the decision variables $x_i(k)$ and $x(k)$ are Gaussian. For the vector~$x(k)$ of the decision variables $x_i(k)$, the vector~$\mu(k)$ of the means $\mu_i(k)$, $1\leq i\leq N$ under $H_1$ (respectively, $H_0$) and the covariance $S_{\mu}(k)$ under either hypotheses are:
\begin{eqnarray}
\label{eqn:vectormu}
\mu(k) &=& {\mathbb E} \left[x(k)\left|H_1\right.\right]=-{\mathbb E} \left[x(k)\left|H_0\right.\right]=\left(\mu_1(k)\,\,\mu_2(k)\,\,\cdots\,\,\mu_N(k)\right)^\top\\
\label{eqn:covariancemu}
S_{\mu}(k)&=&\textrm{Cov}\left(x(k)\right). \nonumber
\end{eqnarray}
We let the diagonal elements of $S_{\mu}(k)$ be
\begin{equation}
\label{eqn:diagSmu}
\sigma_i^2(k)=\left[S_{\mu}(k)\right]_{ii}.
\end{equation}
\mypar{Weight sequences $\gamma_k^1$ and $\gamma_k^2$} Comparing~\eqref{eqn-running-cons-sensor-ii} with~\eqref{eqn:consensus+innovations},  the consensus and innovations weights are
\[
\gamma_k^1=\frac{k}{k+1} \alpha_k \,\,\,\,\,\mathrm{and}\,\,\,\,\,\gamma_k^2=\frac{1}{k+1}.
\]
 Due to the \emph{communication} noise,  the $\left\{\alpha_k\right\}$ have to be diminishing, i.e., $\alpha_k \rightarrow 0$, as pointed out in~\cite{SoummyaChannelNoise,SoummyaEst,Soummya-Detection-Noise}. The design of the $\left\{\alpha_k\right\}$ will be key to the distributed detector achieving exponential decay rate of the error probability:
 a small and fast-decaying $\alpha_k$ injects  low communication noise in the decision variable $x_i(k)$, but limits the information flow among neighbors (insufficient averaging).
 We will show that, for \emph{appropriately designed constants} $a,b_0>0$,
 \begin{equation*}
 \label{eqn_alpha_k}
 \alpha_k = \frac{b_0}{a+k},
 \end{equation*}
  balances these two opposing effects---communication noise and information flow. As detailed in Section \ref{section_perf_analysis}, a large $b_0$ yields larger noise injection but also greater inter-sensor averaging.

We compare the consensus+innovations distributed detector~\eqref{eqn-running-cons-sensor-i} with the distributed detectors in~\cite{running-consensus-detection} and~\cite{Soummya-Detection-Noise}. The detector in~\cite{running-consensus-detection}, referred to as running consensus, uses constant, non-decaying weights $\alpha_k = \alpha$, which is not suitable for noisy communication. To account for communication noise, references~\cite{SoummyaEst,Soummya-Detection-Noise} propose mixed time scale, stochastic approximation type algorithms. In particular, for detection,  \cite{Soummya-Detection-Noise} proposes the
$\mathcal{MD}$~algorithm for the generic case of different signal-to-noise ratios
($\textbf{SNR}$) at different sensors and the single time scale $\mathcal{SD}$ algorithm when the $\textbf{SNR}$ is the same at all sensors. The algorithm~$\mathcal{MD}$ uses the weight sequences (see eqn.~(13), \cite{Soummya-Detection-Noise})
\begin{eqnarray}
\label{eqn:MDweights}
\mathrm{Consensus:}\:\:\:\:&&\{\frac{c_1}{k^\tau}\},\,\,\,\, \tau \in (0.5;1),\,\,\,c_1>0\\
\label{eqn:MDweights-2}
\mathrm{Innovations:}\:\:\:\:&&\left\{\frac{c_2}{k}\right\}, \,\,\,\,c_2>0
\end{eqnarray}
The algorithm $\mathcal{SD}$ uses the same weights for both consensus and innovations (see eqn.~(53) in \cite{Soummya-Detection-Noise}) equal to~\eqref{eqn:MDweights} and~\eqref{eqn:MDweights-2} with~$\tau=1$.
  In contrast with~\cite{Soummya-Detection-Noise}, we propose the single time scale algorithm \eqref{eqn-running-cons-sensor-i} regardless if the local $\textbf{SNR}$ are mutually different or not. A major contribution here with respect to~\cite{Soummya-Detection-Noise} is to show that the single time scale algorithm~\eqref{eqn-running-cons-sensor-i} yields better asymptotic detection performance than the mixed time scale algorithm $\mathcal{MD}$. Our algorithm~\eqref{eqn-running-cons-sensor-i} and $\mathcal{SD}$ are both single time scale. The main differences between~\eqref{eqn-running-cons-sensor-i} and $\mathcal{SD}$ are that algorithm~\eqref{eqn-running-cons-sensor-i}:
  \begin{inparaenum}[1)]
  \item incorporates~$a$ in~\eqref{eqn_alpha_k}; and
  \item optimizes the parameter $b_0$.
  \end{inparaenum}
We will show that~\eqref{eqn-running-cons-sensor-i} exhibits under appropriate structural conditions exponential rate of decay of the probability of error at every sensor, while $\mathcal{MD}$ in~\cite{Soummya-Detection-Noise} is sub exponential;
$\mathcal{SD}$ is shown to be exponential, but only when the sensors are identical, i.e., all operate under the same $\textbf{SNR}$.

\subsection{Signal-to-noise ratios ($\textbf{SNR}$)}
\label{subsect:SNRs}
We define, for future reference, the following relevant $\textbf{SNR}$s.

1. \emph{Sensing $\textbf{SNR}$}:
\vspace{-6mm}
\begin{eqnarray}
\label{eqn-sensing-SNR}
\mathrm{global:}&\:& {\bf SSNR}:=(m_1-m_0)S_{\zeta}^{-1}(m_1-m_0)\\
\mathrm{local:}&\:& {\bf SSNR}_i:=\frac{\left( [m_1-m_0]_i\right)^2}{\left[ S_{\zeta}\right]_{ii}}.
\end{eqnarray}

2. \emph{Detector $\textbf{SNR}$}, for a generic detector $\mathrm{gen}$ (either
centralized, isolated, or distributed):
\begin{equation*}
\label{eqn:DSNRgen}
{\bf DSNR}_{\mathrm{gen}}(k): = \frac{\mathbb{E}^2 \left[ \mathcal{D}_{\mathrm{gen}}(k)\,| \,H_1 \right]} { \mathrm{Var}
\left( \mathcal{D}_{\mathrm{gen}}(k)\,|\,H_1 \right)},
\end{equation*}
where $\mathcal{D}_{\mathrm{gen}}(k)$ is the detector decision variable. We denote
the decision $\textbf{SNR}$  for the centralized, isolated, and distributed detector, respectively,
by ${\bf DSNR}(k)$, ${\bf DSNR}_{\mathrm{i}}(k)$, and ${\bf DSNR}_{\mathrm{dis},i}(k)$.
 We will see later how detector \textbf{SNR} determines the
 error probability for Gaussian detectors; see ahead \eqref{eqn:probabilityerror}.

3. \emph{Communication $\textbf{SNR}$} is a quantity that accounts
for the \emph{communication noise} and plays a role only with the distributed detector.
We define the communication $\textbf{SNR}$ (per sensor) by:
 \begin{equation}
      \label{eqn-G-c}
      {\bf{CSNR}} = \frac{ \left( \frac{1}{N}{\bf SSNR} \right)^2}{\|S_v\|},
      \end{equation}
where we recall the \emph{communication} noise covariance $S_v=\mathrm{Cov}(v(k))$.

      \mypar{Remark} We give a hint why ${\bf CSNR}$ is defined as \eqref{eqn-G-c}
      and why it plays a
       significant role in assessing distributed detection performance.
      With our distributed detector, sensors communicate their
      local decision variables $x_i(k)$; $x_i(k)$ is a local approximation
      of the (scaled) centralized decision variable $\frac{1}{N} \mathcal{D}(k)$ (see \eqref{eqn-d-k}).
     The mean of $x_i(k)$ under $H_1$, for large $k$, is close to the mean of $\frac{1}{N} \mathcal{D}(k)$, equal to
     $\frac{1}{2N} {\bf SSNR}$ (as will be shown); the variance of $x_i(k)$, as will be shown, vanishes at rate $1/k$. Hence, ${{\bf CSNR}}$
     describes how well, in a sense, the signal $x_i(k)$ competes against the noise $v_i(k)$ in communication, for large~$k$.
     We define also the communication gain
      as the ratio of the communication SNR and
      the average (across sensors) sensing SNR\footnote{Note that
      \textbf{CSNR} and \textbf{SSNR} are not independent quantities here; larger \textbf{SSNR} means larger \textbf{CSNR}.}:
      \begin{equation}
      \label{eqn-G-comm}
      {\bf G}_c = \frac{\bf{CSNR}}{\frac{1}{N}{\bf SSNR}} = \frac{ \frac{1}{N}{\bf SSNR} }{\|S_v\|}.
      \end{equation}
      For future reference, we introduce here the following two constants that we will need when assessing distributed detection performance:
\begin{eqnarray}
\label{eqn_c_mu}
c_{\mu} := \frac{2\sqrt{N}\,\|m_{\eta}^{(1)}\|}{\frac{1}{N} {\bf SSNR}}, \:\:
c_{\sigma} := \frac{\|S_{\eta}\|}{\frac{1}{N}  {\bf SSNR}}.
\end{eqnarray}
\vspace{-4mm}
\section{Modeling assumptions and Preliminary results}
\label{section-modeling}
In this Section, we establish
our underlying assumptions and present the asymptotic performance of the isolated sensor
and centralized detectors. These will be a prelude to our main results on
the asymptotic performance of the consensus+innovations detector in~\eqref{eqn-running-cons-sensor-i}
given in Section~\ref{section_perf_analysis} and proven in Section~\ref{section-proofs} and Appendix~\ref{append-beta}.
\vspace{-4mm}
\subsection{Modeling assumptions}
\label{subsect-observ-model}
\label{subsec:modeling}
As mentioned in Section~\ref{section-prob-model}, the noises are zero mean Gaussian spatially correlated but  temporally independent sequences. In Section~\ref{section-extensions}, we will consider the case where the noises are not Gaussian.
\begin{assumption}[Gaussian noises]
\label{assump:gaussnoises}
The \emph{sensing} and the \emph{communication} noises $\zeta_i(k)$ and $\nu_{ij}(k)$ are zero mean, Gaussian spatially correlated and temporally independent noises, and independent of each~other:
\begin{eqnarray*}
\label{eqn:sensingnoise}
\zeta(k) &\sim&\mathcal{N}\left(0,S_{\zeta}\right),\:\:
v(k) \sim \mathcal{N}\left(0,S_v\right),
\end{eqnarray*}
where the vector $v(k)$ is defined in~\eqref{eqn:commnoise-1}
 and $S_{\zeta}$ and $S_v$ are assumed to be positive definite. 
\end{assumption}
 In distributed processing, the ability for the sensors or agents to cooperate is fundamental; this is captured by the connectivity of the network.
\begin{assumption}[Network connectedness]
\label{assump:networkconnectedness}
The network $\mathcal{G}=(\mathcal{V},\mathcal{E})$ is connected.
\end{assumption}
As it is well known, a necessary and sufficient condition for connectedness is
$
\lambda_2(\mathcal{L})>0,
$
i.e., the algebraic connectivity of the network is strictly positive.
The next assumption is on the weight sequence~$\left\{\alpha_k\right\}$.
\begin{assumption}[Weight sequence]
\label{assump:alphak}
The weight sequence $\left\{\alpha_k\right\}$ is:
\begin{equation*}
\label{eqn:alphak}
\alpha_k=\frac{a}{b_0+k},
\end{equation*}
where the constants~$a$ and~$b_0$ satisfy
\begin{eqnarray}
\label{eqn:ab-1}
a   \geq b_0\lambda_N(\mathcal{L})>0,\:\:
b_0 > \max \left\{ 0, \frac{ c_{\mu}-1}{\lambda_2(\mathcal{L})} \right\}.
\end{eqnarray}
\end{assumption}
The role of these conditions will become clear when we state our main result, Theorem~\ref{theorem-rate}, in Section~\ref{section_perf_analysis}. Recall the sensing $\textbf{SNR}$s in~\eqref{eqn-sensing-SNR}. We make the following assumption on $\bf SSNR$.
\begin{assumption}
  \label{def:assump-glob-detect}
  $
  \mathbf{SSNR}>0.
  $
  \end{assumption}
Note that Assumption \ref{def:assump-glob-detect} is equivalent to having different mean vectors, $m_1 \neq m_0$. To obtain certain specialized results, we will assume a stronger assumption than Assumption~\ref{def:assump-glob-detect}.
\begin{assumption}[Equal local sensing $\textbf{SNR}$s] 
  \label{assump:loc-det}
 $\mathbf{SSNR}_i=  \mathbf{SSNR}_j>0,\,\,\, \forall i \neq j.$
  \end{assumption}
%
%
\vspace{-4mm}
\subsection{Asymptotic performance and Chernoff information}
\label{subsec:snr}
For Gaussian decision variables like for the three detectors (isolated sensor, centralized, and consensus+innovations distributed),
   and equal prior probabilities, the probability of error $P^e(k)$ is given by
\begin{equation}
\label{eqn:probabilityerror}
P^e(k) = \mathcal{Q} \left( \sqrt{\mathbf{DSNR}_{\mathrm{gen}}(k)} \right).
\end{equation}
where~$\mathcal{Q}$ is the $Q$-function in~\eqref{eqn:qfunction} and $\mathbf{DSNR}_{\mathrm{gen}}(k)$ is the generic detector $\textbf{SNR}$.

To determine the exponential decay rate of the error probability, we recall the bounds~\cite{Harry-van-Trees-book}
\begin{equation}
\label{eqn-tail-gauss}
  \frac{t}{1+t^2} e^{-t^2/2} \leq 2 \pi\, \mathcal{Q}(t) \leq \frac{1}{t} e^{-t^2/2},\:\: t>0,
\end{equation}
We apply these bounds to~\eqref{eqn:probabilityerror}. Taking the logarithm and dividing by $k$,  the $\limsup$  of the right hand side (rhs) inequality and the $\liminf$ of the left hand side (lhs) inequality in~\eqref{eqn-tail-gauss} lead~to
\footnote{Eqn.~\eqref{eqn-rate-p-e-i} holds because $\mathbf{DSNR}_{\mathrm{gen}}(k)$
 is strictly positive (for large $k$) and can grow at most as $\sim k$ for either centralized, isolated, or distributed detectors (as will be shown).}:
\begin{eqnarray}
\label{eqn-limsup}
\limsup_{k \rightarrow \infty} -\frac{1}{k}
\log P^e(k) &\leq& \limsup_{k \rightarrow \infty} \frac{1}{2k} \mathbf{DSNR}_{\mathrm{gen}}(k)\\
\label{eqn-rate-p-e-i}
\liminf_{k \rightarrow \infty} -\frac{1}{k} \log P^e(k) &\geq&
\liminf_{k \rightarrow \infty} \frac{1}{2k} \mathbf{DSNR}_{\mathrm{gen}}(k).
\end{eqnarray}
If the $\limsup$ in~\eqref{eqn-limsup} is zero, we have two possibilities:
\begin{inparaenum}[1)]
  \item the error probability $P^e(k)$ decays to zero slower than exponentially in $k$; or
  \item $P^e(k)$ does not converge to zero at all.
\end{inparaenum}
Intuitively, large $\mathrm{mean}$ and fast-shrinking $\mathrm{variance}$ of the decision variables $\mathcal{D}_{\mathrm{gen}}(k)$ increase $\mathbf{DSNR}_{\mathrm{gen}}(k)$ and hence yield good detection performance.

The detector for the isolated sensor~$i$ given by equation~\eqref{eqn:isolateddetector-2} and the centralized detector in~\eqref{eqn-d-k} are the optimal minimum probability of error detectors (when prior probabilities are equal, the threshold $\tau(k)=0$.) For these detectors, the $\limsup$ and the $\liminf$ in~\eqref{eqn-limsup} and~\eqref{eqn-rate-p-e-i} actually coincide, i.e., the sequence
$
\frac{1}{2k} \mathbf{DSNR}_{\mathrm{gen}}(k)
$
has a limit and
\begin{eqnarray}
  \label{eqn-c-tot-gen}
  C_{\mathrm{gen}}=\lim_{k \rightarrow \infty} -\frac{1}{k} \log P^e_{\mathrm{gen}}(k) &=&
   \lim_{k \rightarrow \infty} \frac{\mathbf{DSNR}_{\mathrm{gen}}(k)}{2k}.
  \end{eqnarray}
These detectors maximize the exponential decay rate of the probability of error for the corresponding problems; this optimal exponential decay rate is the Chernoff information $C_{\mathrm{gen}}$ as indicated in~\eqref{eqn-c-tot-gen}.
For the isolated sensor~$i$ detector and the centralized detector,
it is easily shown that their detectors $\textbf{SNR}$ $\mathbf{DSNR}_{\mathrm{i}}(k)$ and $\mathbf{DSNR}(k)$ are given by\footnote{We
will see later that $\mathbf{DSNR}_{\mathrm{i,dis}}(k)$ grows also at rate $k$ with our distributed detector~\eqref{eqn-running-cons-sensor-i}.}
\begin{eqnarray}
\label{eqn:DSNRisolated}
\mathbf{DSNR}_{\mathrm{i}}(k)&=&\frac{k}{4}{\bf SSNR}_i, \:\:%
\mathbf{DSNR}_{\mathrm{\phantom{i}}}(k)=\frac{k}{4}{\bf SSNR}.
\end{eqnarray}
The Chernoff information for the isolated and the centralized optimal detectors are then:
\begin{eqnarray}
 \label{eqn-c-isolated}
  C_i=\lim_{k \rightarrow \infty} -\frac{1}{k} \log P^e_{\mathrm{i}\phantom{\mathrm{cen}}}(k) &=&
   \lim_{k \rightarrow \infty} \frac{\mathbf{DSNR}_{\mathrm{i}}(k)}{2k}= \frac{1}{8}\mathbf{SSNR_{\mathrm{i}}}\\
  \label{eqn-c-tot}
  C_{\phantom{\mathrm{i}}}=\lim_{k \rightarrow \infty} -\frac{1}{k} \log P^e_{\mathrm{cen}\phantom{\mathrm{i}}}(k) &=&
   \lim_{k \rightarrow \infty} \frac{\mathbf{DSNR}_{\phantom{\mathrm{i}}}(k)}{2k}= \frac{1}{8}\mathbf{SSNR}_{\phantom{\mathrm{i}}}.
  \end{eqnarray}
Eqns. \eqref{eqn-c-isolated} and \eqref{eqn-c-tot}
 justify the following definition of the global and local detectability, after
 which we relate global and local detectability to sensing (global and local) $\textbf{SNR}$s by a simple, but important fact.
 \begin{definition}
 \label{def:global-local-detect}
 The network $\mathcal{G} = (\mathcal{V}, \mathcal{E})$ is globally detectable if the probability of error~$P^e(k)$ of
 the optimal centralized detector decays exponentially fast.
 The sensor $i \in \mathcal{V}$ is locally detectable if the  probability of error~$P^e_{\mathrm{i}}(k)$ of its optimal detector in isolation decays exponentially fast.
 \end{definition}
 \begin{fact}
 \label{fact:global-det-local-det}
 The network $\mathcal{G} = (\mathcal{V}, \mathcal{E})$ is globally detectable if and only if
 ${\bf SSNR}>0$. The sensor $i \in \mathcal{V}$ is locally detectable if and only if
  ${\bf SSNR}_i>0$.
 \end{fact}
 Clearly, a network can be globally detectable but have many (or most) sensors that are not locally detectable; however, at least one sensor needs to be locally detectable so that global detectability holds. Our goal is to carry out a similar asymptotic performance analysis for the distributed consensus+innovations detector in~\eqref{eqn-running-cons-sensor-i} or in vector form in~\eqref{eqn_recursive_algorithm}.
Because $x_i(k)$ is Gaussian, relations~\eqref{eqn-rate-p-e-i} and~\eqref{eqn-limsup} still apply. Recall the mean and variance of the decision variable $x_i(k)$ under $H_1$, $\mu_i(k)$ and $\sigma_i^2(k)$, given by~\eqref{eqn:vectormu} and~\eqref{eqn:diagSmu}. The distributed detector $\textbf{SNR}$, $\mathbf{DSNR}_{\mathrm{dis,i}}(k)$, at sensor~$i$ at time~$k$ is then:
\begin{equation*}
\label{eqn:DSNR-2}
\mathbf{DSNR}_{\mathrm{dis,i}}(k)=\frac{\left(\mu_i(k)\right)^2}{\sigma_i^2(k)}.
\end{equation*}
We obtain the moments $\mu_i(k)$ and $\sigma_i^2(k)$ of $x_i(k)$ and their asymptotic values  in the next
Section~\ref{section_perf_analysis} by analyzing the distributed algorithm in~\eqref{eqn-running-cons-sensor-i}. In contrast with the centralized and isolated detectors, these statistics of the decision variable~$x_i(k)$ of the~$i$th sensor are affected by the communication noise $\nu_{ij}(k)$, see equation~\eqref{eqn:commnoise-0}, through $v_i(k)$ and $v(k)$ in~\eqref{eqn:commnoise-1} and~\eqref{eqn:commnoise-1}. We will see that, besides $\mathbf{DSNR}_{\mathrm{dis,i}}(k)$, we need to account for the impact of the~$\bf{G_c}$ given in~\eqref{eqn-G-comm}.

%
\vspace{-5mm}
\section{Consensus+Innovations distributed detection: Performance analysis}
\label{section_perf_analysis}
Subsection \ref{subs:exp-decay} studies
the exponential decay of our consensus+innovations detector in \eqref{eqn-running-cons-sensor-i}, subsection \ref{subsect-order-optim} addresses the optimality
  of the weight sequence $\{\alpha_k\}$, and
  subsection \ref{subsect:payoff} addresses the potential payoff of distributed detection arising from noisy cooperation among sensors.
\subsection{Exponential decay of the error probability}
\label{subs:exp-decay}
The next Theorem establishes under reasonable conditions that the probability of error at every sensor of the consensus+innovations distributed detector in~\eqref{eqn-running-cons-sensor-i} decays exponentially fast.
%
%
%
Recall the definitions of $\textbf{SSNR}$  and $\bf{G_c}$ in~\eqref{eqn-sensing-SNR} and~\eqref{eqn-G-comm}, and the constants $c_{\mu}$ and $c_{\sigma}$ in~\eqref{eqn_c_mu}.
\begin{theorem}
\label{theorem-rate}
 Consider the consensus+innovations distributed detector in \eqref{eqn-running-cons-sensor-i} under the
 Assumptions~\ref{assump:gaussnoises}, \ref{assump:networkconnectedness}, \ref{assump:alphak}, and~\ref{def:assump-glob-detect}. Then:
 \begin{enumerate}
 \item The moments $\mu(k)$, $\mu_i(k)$, and $\sigma_i^2(k)$ satisfy:
 \begin{eqnarray}
\label{eqn-prva-za-dokaz}
\mu_{\infty} &:=& \lim_{k \rightarrow \infty} \mu(k) = \left( I + b_0 \,\mathcal{L} \right)^{-1}m_{\eta}^{(1)}\\
\label{eqn-mu-bound-lower}
\lim_{k \rightarrow \infty} \mu_i(k) &\geq&
\frac{1}{2 N} {\bf SSNR}
\left( 1 - \frac{\sqrt{N}}{1+ b_0 \lambda_2(\mathcal{L})}c_{\mu} \right)\\
\label{eqn-sigma-2-k}
\limsup_{k \rightarrow \infty} k\,\sigma_i^2(k) &\leq&
\frac{1}{N^2} {\bf SSNR}
\left( 1 + 3 \frac{N}{1+b_0 \lambda_2(\mathcal{L})} c_{\sigma} + \frac{N b_0^2}{\bf{G_c}}  \right).
\end{eqnarray}
\item The exponential decay rate of the  error probability $P^e_{\mathrm{dis},i}(k)$ at every sensor~$i$ satisfies:
\begin{eqnarray}
\label{eqn_mu_i_k_liminf}
\liminf_{k \rightarrow \infty} -\frac{1}{k} \log P^e_{\mathrm{dis},i}(k)
\geq \frac{1}{8}{\bf SSNR} \frac{\left(1 -
\frac{\sqrt{N}}{1+b_0 \lambda_2(\mathcal{L})}c_{\mu} \right)^2}{1 + 3 \frac{N}{1+b_0 \lambda_2(\mathcal{L})} c_{\sigma} + \frac{N b_0^2}{\bf{G_c}}}.
\end{eqnarray}
\end{enumerate}
\end{theorem}
Before proving the Theorem, which we carry out in Section~\ref{section-proofs}, we analyze how the bound on the rhs of~\eqref{eqn_mu_i_k_liminf} depends on the different $\textbf{SNR}$s, on the network connectivity $\lambda_2(\mathcal{L})$, and on the parameter $b_0$ of the weight sequence $\left\{\alpha_k\right\}$. The discussion is summarized in the following
 five remarks on Theorem~\ref{theorem-rate}.
\textbf{1. Exponential decay of the error probability $P^e_{\mathrm{dis,i}}$.}
Under global detectability and connectedness, Theorem~\ref{theorem-rate} states that the error probability $P^e_{\mathrm{dis,i}}$ at \emph{every}  sensor~$i$ decays
 exponentially to zero \emph{even if sensor $i$ is (in isolation) not detectable} (${\bf SSNR}_i=0$,)
 and even when the communication links are very noisy (${\bf{G_c}}>0$ but small.) This feature of the distributed
   detector~\eqref{eqn-running-cons-sensor-i} significantly
   improves over existing work like $\mathcal{MD}$ in~\cite{Soummya-Detection-Noise}.
   Namely, we prove in Appendix \ref{subsect-soummya} that $\mathcal{MD}$ achieves only a sub exponential
   decay rate of order $e^{-c\,k^\tau}$, $\tau<1$, $c>0$, of the error probability,
   irrespective of~${\bf{G_c}}$.

\textbf{2. Effect of ${{\bf G_c}}$.} The bound on the rhs of~\eqref{eqn_mu_i_k_liminf} shows quantitatively
   that higher ${{\bf G_c}}$ leads to better detection, confirming the qualitative discussion in the Remark below~\eqref{eqn-G-c}.

\textbf{3. Effect of the network connectivity $\lambda_2(\mathcal{L})$.} Theorem \ref{theorem-rate} shows that
the network connectivity plays a role in the detection performance
 through the algebraic connectivity $\lambda_2(\mathcal{L})$.
Larger values of $\lambda_2(\mathcal{L})$, which allow for faster averaging, increase the bound~\eqref{eqn_mu_i_k_liminf} yielding faster decay rate for the  error probability.

\textbf{4. Tradeoff: Communication noise vs.~information flow.} With optimal centralized detection (that corresponds to a fully connected network and no additive communication noise)
we have that, for all $i$, $P^e_i(k)\equiv P^e(k)$, and:
$
\lim_{k \rightarrow \infty} -\frac{1}{k} \log P^e_i(k) = \frac{1}{8}{\bf SSNR}.
$
Then, from \eqref{eqn_mu_i_k_liminf}, all the three terms:
\begin{eqnarray}
\label{eqn-jedan-minus}
& \frac{\sqrt{N}}{1+b_0 \lambda_2(\mathcal{L})}c_{\mu} &\\
\label{eqn-3}
& 3 \frac{N}{1+b_0 \lambda_2(\mathcal{L})} c_{\sigma} & \\
\label{eqn-snr-deter}
& \frac{N b_0^2}{\bf{G_c}}=\frac{N}{\frac{\bf{G_c}}{b_0^2}} %
\end{eqnarray}
decrease the bound and so they quantify the decrease in performance of the distributed detector with respect to the centralized detector.
This decrease comes from two effects: 1) communication noise; and 2)
insufficient information flow.

From~\eqref{eqn-jedan-minus}--\eqref{eqn-snr-deter},
 we can see how the parameter~$b_0$ affects in opposing ways these two effects:
 The terms~\eqref{eqn-jedan-minus} and~\eqref{eqn-3} relate to the information flow, while the term~\eqref{eqn-snr-deter}
  is due to communication noise. We see that the net effect of increasing~$b_0$ is to increase the effective algebraic connectivity ($b_0$ multiplies $\lambda_2(\mathcal{L})$), increasing~\eqref{eqn-jedan-minus} and~\eqref{eqn-3}; on the other hand, it reduces~$\textbf{CSNR}$ as seen from~\eqref{eqn-snr-deter}.

  The weight choice $\alpha_k$ in~\eqref{eqn_alpha_k} optimally balances these two effects if we tune the parameter $b_0$ to maximize the right hand side in~\eqref{eqn_mu_i_k_liminf}. This is a scalar optimization problem in $b_0$ and can be easily numerically performed. Lemma \ref{lemma-b-opt} find the optimal $b_0$ in closed form for a simplified~case.
%
%
%

\textbf{5. Tradeoff: Bias-variance.} Theorem \ref{theorem-rate} reveals a certain bias-variance tradeoff.
Ideally, we would like the bias-free decision variables:
\begin{equation}
\label{eqn-bias-free}
\mu_{\infty} = \left( \frac{1}{2 N} {{\bf SSNR}} \right)\,1,
\end{equation}
where~$1$ is the vector of ones; i.e., all sensors should have as asymptotic decision variable the asymptotic centralized decision variable. That is, we want the mean of the decision variable at each sensor to converge to the expected value of the centralized decision variable $\mathcal{D}(k)$ (divided by $1/N$.)  Our algorithm \eqref{eqn-running-cons-sensor-i} introduces a bias (see \eqref{eqn-prva-za-dokaz} and \eqref{eqn-mu-bound-lower}), but, on the other hand, it decreases
 the variance at the optimal rate $1/k$. In contrast, $\mathcal{MD}$ in \cite{Soummya-Detection-Noise}
  does not have the bias, but it decreases the variance at a slower rate.
  Compared to $\mathcal{MD}$, our algorithm \eqref{eqn-running-cons-sensor-i}
  better resolves the bias-variance tradeoff in terms of the detection performance;
   algorithm~\eqref{eqn-running-cons-sensor-i} decays the error probability
   exponentially, while $\mathcal{MD}$ decays it sub exponentially.
%
%
%
%
%
%
%
We now consider a special case where all sensors are identical, or, better said, they operate under the same $\textrm{SSNR}_i$, i.e.,
 Assumption \ref{assump:loc-det} holds. Theorem \ref{theorem-rate} takes a simplified form, where $\mu_{\infty}$ becomes bias free, as in \eqref{eqn-bias-free}. Further, $c_{\mu}=c_{\sigma}=1$ and second condition in~\eqref{eqn:ab-1} becomes $b_0>0$; it can also be shown (details omitted) that the factor~$3$ in \eqref{eqn-3} reduces to~$1$. The simplified Theorem~\ref{theorem-rate}~follows.
\begin{theorem}[Asymptotic performance: Identical sensors]
\label{new-thm}
Let Assumptions \ref{assump:gaussnoises} through \ref{assump:alphak} and~\ref{assump:loc-det} hold. Then, the
exponential decay rate of the error probability at each sensor $i$ satisfies:
\begin{eqnarray}
\label{eqn_mu_i_k_liminf-identical}
\liminf_{k \rightarrow \infty} -\frac{1}{k} \log P^e_{\mathrm{dis},i}(k)
\label{eqn-tighter}
&\geq& \frac{1}{8}{\bf SSNR} \frac{1}{ \left( 1 + \frac{N}{1+b_0 \lambda_2(\mathcal{L})} + \frac{N b_0^2}{\bf{G_c}}  \right)}
\\
\label{eqn-looser-bound}
&\geq& \frac{1}{8}{\bf SSNR}  \frac{1}{ \left( 1 + \frac{N}{b_0 \lambda_2(\mathcal{L})} + \frac{N b_0^2}{\bf{G_c}}  \right)}
\end{eqnarray}
\end{theorem}
\subsection{Optimality of the weight sequence $\{\alpha_k\}$}
\label{subsect-order-optim}
\mypar{Order-optimality} We consider the role of the weight sequence~\eqref{eqn_alpha_k}, in particular, we show the optimality of the rate $1/k$. To this end, we consider the distributed detector~\eqref{eqn-running-cons-sensor-i} but modify the weight sequence;
we refer to the modified sequence as $\beta_k$. We find an upper bound
on the decay rate of the error probability when the weight sequence is re-set to $\alpha_k$.
\begin{theorem}
\label{lemma-choice-beta}
Let Assumptions \ref{assump:gaussnoises}---\ref{assump:alphak} and~\ref{assump:loc-det} hold.
 Suppose
that the weight choice $\alpha_k$ in \eqref{eqn_alpha_k}
is replaced by:
\[\beta_k = \frac{b_0}{(a+k^\tau)},\:\: a,b_0>0,\]
where $\tau \geq 0$.
We have:

\begin{equation}
\label{eqn-tightness-new}
 \limsup_{k \rightarrow \infty} - \frac{1}{k} \log P^e_{\mathrm{dis},i}(k) \leq
 \left\{ \begin{array}{lll}
  0  &\mbox{ if $\tau<1$} \\
  \frac{1}{8}{\bf SSNR}_i &\mbox{ if $\tau>1$}\\
  \frac{1}{8}{\bf SSNR} \frac{1}{1 + \frac{N}{1+2 b_0 \lambda_N(\mathcal{L})} + \frac{b_0^2\,\lambda_1(S_v)}{N {\bf SSNR}_i}} &\mbox{ if $\tau=1.$}
       \end{array} \right.
\end{equation}
       %
\end{theorem}
\begin{proof}
See Appendix \ref{append-beta}.
\end{proof}
Two remarks on Theorem~\ref{lemma-choice-beta} are in order.
%

\textbf{1. Order-optimality of $\alpha_k$.} Theorem \ref{lemma-choice-beta} says that the choice
$\alpha_k$ in \eqref{eqn_alpha_k} is the optimal weight choice
in the family of choices $\beta_k = \frac{b_0}{(a+k^\tau)}$, $a,b_0>0,$
parametrized by $\tau \geq 0$. If
$\beta_k$ decays too slowly ($\tau <1$),
then the  error probability converges to
zero at a rate slower than exponential (if at all it converges to zero.) On the other hand, if $\beta_k$
 decays too fast ($\tau >1$), then the  error
 probability does decay to zero exponentially, but the rate
 is no better than the rate of the individual detection, irrespective
 of $\bf{G_c}$.

\textbf{2. Tightness of the bounds in Theorems~\ref{new-thm} and~\ref{lemma-choice-beta}.}
The upper bound~\eqref{eqn-tightness-new} for $\tau=1$ explains the tightness of the lower bound in Theorem \ref{new-thm}
 and the unavoidable simultaneous effects of the communication noise and information flow. The sequence $\{\alpha_k\}$ balances these via the parameter $b_0$.

%
%


%

\mypar{Optimal $b_0^\star$} We now find $b_0^\star$
that optimizes (maximizes) \eqref{eqn-looser-bound};
 we pursue~\eqref{eqn-looser-bound} rather than~\eqref{eqn-tighter}
  as it allows simpler, closed form expressions. Proof of Lemma~6 follows after setting
  the derivative of the denominator of rhs in~(42) to zero and is hence omitted for brevity.
%
%
\begin{lemma}
\label{lemma-b-opt}
Let Assumptions \ref{assump:gaussnoises} through \ref{assump:alphak} and \ref{assump:loc-det} hold. The optimal parameter $b_0^\star$
that maximizes \eqref{eqn-looser-bound} and
the corresponding lower bound on $\liminf_{k \rightarrow \infty} -\frac{1}{k} \log P^e_i(k)$, are, respectively:
\begin{eqnarray}
\label{eqn-b-opt}
b_0^\star  &=&  \frac{{{\bf G_c}}^{1/3} } {  \lambda_2(\mathcal{L})^{1/3}4^{1/3}    }\\
\label{eqn_b_star}
\liminf_{k \rightarrow \infty} -\frac{1}{k} \log P^e_{\mathrm{dis},i}(k) &\geq &
{\frac{1}{8}{\bf SSNR}} \frac{1}{1 +  \frac{N c_0}{ \lambda_2(\mathcal{L})^{2/3} \, {\bf{G_c}}^{1/3}}},
\end{eqnarray}
where \[c_0=\frac{1}{2}4^{1/3}+4^{-2/3}=\frac{3}{2}(2)^{-1/3}\approx 1.19.\]
\end{lemma}
We use Lemma \ref{lemma-b-opt} to compare the distributed detector with the optimal centralized detector and the optimal single sensor detector. In the very high $\bf{G_c}$ regime (weak communication noise), when ${\bf{G_c}} \rightarrow \infty$,
 the distributed detector (at all sensors) achieves the asymptotic performance of the optimal centralized detector.
 On the other hand, when ${\bf{G_c}}$ decreases, at some point, the rhs
  in~\eqref{eqn_b_star} falls below $\frac{1}{8}{\bf SSNR}_i$ and the
  distributed detector~\eqref{eqn-running-cons-sensor-i} at sensor~$i$ becomes worse than if sensor~$i$ worked in isolation. The discussion is formalized in the next Subsection that considers when sensors should cooperate.

  \subsection{Communication payoff}
  \label{subsect:payoff}
  Eqn.~\eqref{eqn_b_star} under low~$bf{G_c}$
   raises the issue whether a sensor~$i$ should cooperate with its neighbors or not. We next formalize communication payoff.
\begin{definition}[Communication payoff]
\label{def-1}
The network $\mathcal{G}=(\mathcal{V}, \mathcal{E})$ achieves communication payoff if:
\begin{equation*}
\label{eqn-def-pay}
\min_{i=1,...,N} \left\{ \liminf_{k \rightarrow \infty} -\frac{1}{k} \log P^e_{\mathrm{dis},i}(k) \right\} \geq
\max_{i=1,...,N} \left\{\frac{1}{8}{\bf SSNR}_i\right\}.
\end{equation*}
\end{definition}
Definition~\ref{def-1} says that the network achieves a communication payoff if the distributed detector error performance of the worst sensor is better than the isolated detector error performance for the best sensor without communication. Lemma~\ref{lemma-lower-bound} finds a threshold on the~${{\bf G_c}}$ above which
it does pay off for sensors to communicate with their neighbors. Proofs of Lemma~8 is simple and~is~omitted.
\begin{lemma}
\label{lemma-lower-bound}
Let Assumptions \ref{assump:gaussnoises} through \ref{assump:alphak} and \ref{assump:loc-det} hold.
 Set $b_0$ to the optimal value in \eqref{eqn-b-opt}.
 If
 \[{\bf {G_c}} \geq \left( c_0 \frac{N}{N-1}  \right)^3 \left( \frac{1}{\lambda_2(\mathcal{L})}\right)^2,\]
 then the network achieves the communication payoff in the sense of Definition \ref{def-1}.
\end{lemma}
%
%
%
%
%

\section{Proof of Theorem~\ref{theorem-rate}}
\label{section-proofs}
Subsection~\ref{subsection-set-up-proofs} sets up the analysis
and Subsection~\ref{subsec:prooftheorem} proves Theorem~\ref{theorem-rate}.
\subsection{Solution of the consensus+innovations distributed detector}
\label{subsection-set-up-proofs}
    Define the matrices $\Phi(k,j)$, $k \geq j \geq 1$, as follows:
    \begin{equation*}
    \Phi(k,j):=\left\{ \begin{array}{ll}
  W(k-1)W(k-2)...W(j)  &\mbox{ if $1 \leq j <k$} \\
  I &\mbox{if $j=k.$}
       \end{array} \right.
       \end{equation*}
Then, the solution to the distributed 
    detector~\eqref{eqn_recursive_algorithm} is:
\begin{equation}
\label{eqn_x(k)-equation}
x(k) = \frac{1}{k} \sum_{j=1}^{k} \Phi(k,j) \eta(j) + \frac{1}{k} \sum_{j=1}^{k-1} \left( j\,\alpha_j \right) \Phi(k,j+1)  v(j), \:\:k=1,2,3,...
\end{equation}
 Introduce
\begin{eqnarray*}
\widetilde{W}(k) &:=& W(k)-J\\
\label{eqn-tilde-phi}
\widetilde{\Phi}(k,j) &:=& \widetilde{W}(k-1) \widetilde{W}(k-2) ... \widetilde{W}(j), \,\,k > j \geq 1.
\end{eqnarray*}
 It can be seen that
\[
\widetilde{\Phi}(k,j) = \Phi(k,j)-J.
\]
In consensus, $W(k)\rightarrow J$, where~$J$ is the ideal consensus averaging matrix. The matrix $\widetilde{W}(k)$ and its norm $\| \widetilde{W}(k)\|$ measure, in a sense, the imperfection in the information flow, i.e., how
 far~$\widetilde{W}(k)$ is away from~$0$ and~$W(k)$ from~$J$.
 If~\eqref{eqn:ab-1} holds then it is easy to see that
\begin{eqnarray}
\label{eqn_norm_tilde_W}
\|\widetilde{W}(k)\| & = &1 - \frac{b_0\lambda_2(\mathcal{L})}{a+k} \in \left( 1-\frac{\lambda_2(\mathcal{L})}{\lambda_N(\mathcal{L})},1 \right), \,\, \forall k.
\end{eqnarray}
 We see that the role of~$a$ in~$\alpha_k$ is to be an offset that enables \eqref{eqn_norm_tilde_W} to hold for all~$k$; that is, $a$ reduces $\|\widetilde{W}(k)\|$ for large $b_0$ and small $k$. We will see that $b_0$ is the effective tuning parameter that controls the detection performance. We also comment that the ratio $\frac{\lambda_2(\mathcal{L})}{\lambda_N(\mathcal{L})}$ is maximized by Ramanujan networks, see~\cite{SoummyaRamanujan} for details.
%
%
%
%
\subsection{Proof}
\label{subsec:prooftheorem}
%
%
%
\begin{proof}[Proof of Theorem~\ref{theorem-rate}: Claims~\eqref{eqn-prva-za-dokaz} and~\eqref{eqn-mu-bound-lower}]
We first study the mean of the decision variable $\mu(k)$. It
 evolves according to the following recursion (which can be seen by taking the expectation in~\eqref{eqn_recursive_algorithm}):
 \begin{eqnarray}
 \label{eqn-rec-mu-k}
 \mu(k+1) &=& \frac{k}{k+1}W(k)\mu(k)+\frac{m_{\eta}^{(1)}}{k+1}.
 \end{eqnarray}
Next, we consider the \emph{error} $\epsilon(k)$ of $\mu(k)$ wrt the assumed $\mu_{\infty}$ given in~\eqref{eqn-prva-za-dokaz}:
\begin{equation*}
\epsilon(k):=\mu(k)-\left( I+b_0\,\mathcal{L}\right)^{-1}m_{\eta}^{(1)}.
\end{equation*}
We will show that $\epsilon(k) \rightarrow 0$, which implies \eqref{eqn-prva-za-dokaz}.
Algebraic manipulations show that $\epsilon(k)$ satisfies:
\begin{equation}
\label{eqn-rec-epsilon}
\epsilon(k+1) = \frac{k}{k+1}W(k)\epsilon(k) + \frac{1}{k+1} \Gamma(k) m_{\eta}^{(1)},
\end{equation}
where
\begin{eqnarray*}
\Gamma(k) &=& I - (k+1) \left( I+b_0\,\mathcal{L} \right)^{-1} + k W(k)\left( I+b_0\,\mathcal{L} \right)^{-1}\\
          &=& I - (k+1) \left( I+b_0\,\mathcal{L} \right)^{-1} + k \left( I - \frac{b_0}{a+k}\,\mathcal{L} \right) \left( I+b_0\,\mathcal{L} \right)^{-1}.
\end{eqnarray*}
 Recall the eigendecomposition of the Laplacian in~\eqref{eqn:Laplaciandecomposition}. The matrix $\Gamma(k)$ has the same eigenvectors as $\mathcal{L}$; simple calculations show that the eigenvalue $\lambda_i\left(\Gamma(k) \right) $ that corresponds to the eigenvector $q_i$ is:
 \begin{equation*}
 \lambda_i\left(\Gamma(k) \right) =
 \left\{ \begin{array}{ll}
  0  &\mbox{ if $i=1$} \\
   \frac{b_0\, a\, \lambda_i(\mathcal{L})}{1 + b_0 \lambda_i(\mathcal{L})} \frac{1}{k+a} = :\frac{c_{\Gamma,i}}{k+a} &\mbox{ otherwise.}
       \end{array} \right.
 \end{equation*}
 Then, clearly, for some $c_{\Gamma}>0$,
 \begin{equation}
 \label{eqn-norm-Gamma-k}
 \| \Gamma(k) \| \leq \frac{c_{\Gamma}}{k}.
 \end{equation}
We now decompose $\epsilon(k)$ into the \emph{consensus subspace}, i.e., the component colinear with the vector~$1$, and the component
orthogonal to~$1$:
$
\epsilon(k) = \left( I - J\right) \epsilon(k) + J \epsilon(k) = \left( I - J\right) \epsilon(k)
 + \left( \frac{1}{N} 1^\top \epsilon(k) \right) 1.
$ We show separately that:
 \begin{eqnarray}
 \label{eqn-orthogonal}
 \lim_{k \rightarrow \infty} \left( I - J\right) \epsilon(k) &=& 0 \\
 \label{eqn-colinear}
 \lim_{k \rightarrow \infty}  1^\top \epsilon(k) &=& 0.
 \end{eqnarray}
 Then, \eqref{eqn-orthogonal} and \eqref{eqn-colinear} together imply that
 \begin{equation}
 \label{eqn-total}
 \lim_{k \rightarrow \infty} \epsilon(k) = 0.
 \end{equation}
 We first show \eqref{eqn-colinear}. Multiplying
 \eqref{eqn-rec-epsilon} from the left by $1^\top$, using the orthogonality of the eigenvectors $q_i$, and using the fact that $1^\top W(k) = 1^\top$, we get:
 \begin{equation*}
 1^\top \epsilon(k+1) = \frac{k}{k+1} 1^\top \epsilon(k),
 \end{equation*}
 which implies \eqref{eqn-colinear}.

 We now show \eqref{eqn-orthogonal}. Denote by
 \[
 b:=b_0 \lambda_2(\mathcal{L}).
 \]
Multiplying \eqref{eqn-rec-epsilon}
 from the left by $(I-J)$,
 we get:
 \begin{eqnarray}
\left(I-J \right) \epsilon(k+1)&=& \frac{k}{k+1}
\left(I-J \right) W(k) \epsilon(k) +  \frac{1}{k+1} \left(I-J \right) \Gamma(k)  m_{\eta}^{(1)} \nonumber \\
&=&
\label{eqn-bas-sad-ova}
\frac{k}{k+1}
\widetilde{W}(k) (I-J)\epsilon(k) +  \frac{1}{k+1}  \Gamma(k)  m_{\eta}^{(1)},
\end{eqnarray}
where \eqref{eqn-bas-sad-ova} holds because $J\,W(k)=J$,
$\widetilde{W}(k)J=0$, and
$\left(I-J \right) \Gamma(k)=\Gamma(k)$.
Now, by subadditivity and submultiplicativity of norms, \eqref{eqn-bas-sad-ova} yields:
\begin{eqnarray}
\left\| \left(I-J \right) \epsilon(k+1) \right\|
&\leq& \frac{k}{k+1} \| \widetilde{W}(k)\| \|(I-J)\epsilon(k)\| +
\frac{1}{k+1} \|\Gamma(k)\| \|m_{\eta}^{(1)}\| \nonumber  \\
                  &\leq& \left(1 - \frac{b_0\,\lambda_2(\mathcal{L})}{a+k}  \right) \|(I-J)\epsilon(k)\| +
                  \frac{c_{\Gamma} \|m_{\eta}^{(1)}\|}{k^2} \nonumber  \\
                  &=&   \left(1 - \frac{b}{a+k}  \right) \|(I-J)\epsilon(k)\|  + \frac{c_{\epsilon}}{k^2} \nonumber \\
                  \label{eqn-za-lemmu}
                  &=&   \|(I-J)\epsilon(k)\| - \frac{b}{a+k} \|(I-J)\epsilon(k)\|  + \frac{c_{\epsilon}}{k^2}.
\end{eqnarray}

Before proceeding, we invoke the following  deterministic
variant of a result due to Robbins and Siegmund (Lemma~11,~Chapter 2.2.,~\cite{Polyak}.)
\begin{lemma}[\cite{Polyak}]
\label{lemma-siegmund}
Let $\{u(k)\}$, $\{\rho(k)\}$, and $\{\kappa(k)\}$ be non-negative deterministic (scalar) sequences.
Further, suppose that
\begin{equation*}
u(k+1)  \leq  u(k) - \rho(k) + \kappa(k),\,\,k=1,2,...
\end{equation*}
Suppose that $\sum_{k=1}^\infty \kappa(k) < \infty$. Then: 1)
$\sum_{k=1}^\infty \rho(k)<\infty$; and 2)
 $\lim_{k \rightarrow \infty} u(k) = u^\star$ exists.
\end{lemma}

We apply Lemma~\ref{lemma-siegmund} to~\eqref{eqn-za-lemmu} with
  \begin{eqnarray*}
  u(k) = \|(I-J)\epsilon(k)\| ,\:\:
  \rho(k) = \frac{b}{a+k} \|(I-J)\epsilon (k)\| ,\:\:
  \kappa(k) = \frac{c_{\epsilon}}{k^2}.
  \end{eqnarray*}
   This proves that \[\lim_{k \rightarrow \infty} \|(I-J)\epsilon(k)\| = 0,\]
   i.e., proves \eqref{eqn-orthogonal}. Namely, by Lemma \ref{lemma-siegmund},
     we have that \[\sum_{k=1}^\infty \rho(k) =  \sum_{k=1}^\infty  \frac{b}{a+k} \|(I-J)\epsilon(k)\| < \infty,\]
     which implies that \[\liminf_{k \rightarrow \infty}\|(I-J)\epsilon (k)\| =0 .\]
     Also, by Lemma \ref{lemma-siegmund}, $\lim_{k \rightarrow \infty} u(k)
      = \lim_{k \rightarrow \infty} \|(I-J)\epsilon (k)\| $ exists, and, hence,
      $\lim_{k \rightarrow \infty} \|(I-J)\epsilon (k)\|=0.$
      This completes the proof of \eqref{eqn-total}.

We now prove \eqref{eqn-mu-bound-lower} using \eqref{eqn-prva-za-dokaz}. Note first that
\begin{equation}
\left( I + b_0\,\mathcal{L} \right)^{-1} = Q\, \Lambda(\left( I + b_0\,\mathcal{L} \right)^{-1} )\,Q^\top,
\end{equation}
where \[\Lambda( \left( I + b_0\,\mathcal{L} \right)^{-1} )
= \mathrm{Diag} \left( 1, (1+b_0 \lambda_2(\mathcal{L}))^{-1},...,(1+b_0 \lambda_N(\mathcal{L}))^{-1} \right).\]
Thus, using the fact that $q_1=\frac{1}{\sqrt{N}}1$ and $J=q_1q_1^\top$, the matrix $\left( I + b_0\,\mathcal{L} \right)^{-1}$ decomposes as:
\begin{equation}
\label{eqn-bas-ova}
\left( I + b_0\,\mathcal{L} \right)^{-1} = J + Q \Lambda^\prime Q^\top,
\end{equation}
with $\Lambda^\prime = \mathrm{Diag} \left( 0, (1+b_0 \lambda_2(\mathcal{L}))^{-1},...,(1+b_0 \lambda_N(\mathcal{L}))^{-1} \right).$
Multiplying \eqref{eqn-bas-ova} from the right
by $m_{\eta}^{(1)}$, and using \eqref{eqn_mu_sigma}, we get that the entry $[\mu_{\infty}]_i$ equals
\begin{equation}
[\mu_{\infty}]_i = \frac{1}{2N} {\bf SSNR} + \left[ Q \Lambda^\prime Q^\top m_{\eta}^{(1)} \right]_i,
\end{equation}
Finally, the inequality $\left|  \left[ Q \Lambda^\prime Q^\top m_{\eta}^{(1)} \right]_i \right| \leq \|\Lambda^\prime \|
\|m_{\eta}^{(1)} \| = \frac{1}{1+b_0 \lambda_2(\mathcal{L})} \|m_{\eta}^{(1)}\|
$
yields \eqref{eqn-mu-bound-lower}.
\end{proof}

We now prove~\eqref{eqn-sigma-2-k}; we use the following auxiliary result.
\begin{lemma}
\label{lemma-aux}
Denote by:
\begin{eqnarray*}
\mathcal{Z}(k) &:=& \frac{1}{k} \sum_{j=1}^k \| \widetilde{\Phi}(k+1,j) \| = \frac{1}{k} \sum_{j=1}^k\Pi_{j=1}^k \left( 1-\frac{b}{a+j}\right),\:
\mathcal{Z}^\star := \limsup_{k \rightarrow \infty} \mathcal{Z}(k) \\
{\chi}(k) &:=& \frac{1}{k} \sum_{j=1}^{k-1} \left( j^2 \alpha_j^2\right),\:
{\chi}^\star := \lim_{k \rightarrow \infty} {\chi}(k).
\end{eqnarray*}
Then, for all $i$, the following holds:
\begin{eqnarray}
\label{eqn-sigma-2-k-2}
\limsup_{k \rightarrow \infty} k\,\sigma_i^2(k) &\leq&
\frac{1}{N^2} {\bf{SSNR}} + 3 \|S_{\eta}\| \mathcal{Z}^\star
 + \|S_v\| \chi^\star.
\end{eqnarray}
Moreover, we have:
\begin{eqnarray}
\label{eqn-z-star}
\mathcal{Z}^\star &\leq& \frac{1}{b+1} = \frac{1}{\lambda_2(\mathcal{L})\,b_0+1}\\
\label{eqn-chi-star}
\chi^\star  &=&  b_0^2.
 \end{eqnarray}
\end{lemma}

%
%
%
%
%
%
\begin{proof}[Proof of Lemma \ref{lemma-aux}] Consider \eqref{eqn_x(k)-equation}. Using the
independence of $\eta(j)$ and $\eta(k)$, $k \neq j$,
and the independence of $\eta(k)$ and $v(j)$, for all $k,j$,
and using the equality $\Phi(k,j)=\widetilde{\Phi}(k,j)+J$, we have:
\begin{eqnarray*}
\sigma_i^2(k) &=&   \frac{1}{k^2} \sum_{j=1}^k \mathrm{Var} \left( e_i^\top \Phi(k,j) \eta(j)   \right)
+ \frac{1}{k^2} \sum_{j=1}^{k-1} \left( \alpha_j\,j\right)^2 \, \mathrm{Var} \left( e_i^\top \Phi(k,j+1) v(j) \right) \\
\label{eqn-sigma-temp}
&=& \frac{1}{k^2} \,k\, \left( e_i^\top J S_{\eta} J e_i \right) + \frac{1}{k^2} e_i^\top (I-J) S_{\eta} (I-J)e_i
+ \frac{2}{k^2}  \sum_{j=1}^{k-1} e_i^\top J S_{\eta} \widetilde{\Phi} (k,j) ^\top e_i\\
&+& \frac{1}{k^2} \sum_{j=1}^{k-1} e_i^\top \widetilde{\Phi} (k,j) S_{\eta} \widetilde{\Phi}(k,j)^\top e_i
+ \frac{1}{k^2} \sum_{j=1}^{k-1} \left( \alpha_j \,j\right)^2 \left(  e_i^\top \Phi(k,j+1) S_v \Phi(k,j+1)^\top e_i \right).
\end{eqnarray*}
Straightforward algebra shows:
\begin{equation*}
\label{eqn-st-al}
e_i^\top J S_{\eta} J e_i = \frac{1}{N^2} {\bf SSNR}.
\end{equation*}
We next bound from above the
quantity $k \sigma_i^2(k)$, using \eqref{eqn-st-al} and
the following norm arguments: 1) $\|A B\| \leq \|A\|\,\|B\|$; 2)
$\|A b\| \leq \|A\| \, \|b\|$, for square matrices $A$ and $B$, and a vector $b$; 3)
 $\|e_i\|=1$; 4) $\|\Phi(k,j+1)\|=1$. (The latter claim is because
 $\|\Phi(k,j+1)\|$ is doubly stochastic.) The bound on
 $k \sigma_i^2(k) $ is as follows:
\begin{eqnarray}
k \sigma_i^2(k)  &\leq& \frac{1}{N^2} {\bf SSNR}
+ \frac{2}{k} \|S_{\eta}\|  \sum_{j=1}^{k-1} \|\widetilde{\Phi}(k,j)\| \nonumber \\
&+& \frac{1}{k} \|S_{\eta} \| \sum_{j=1}^{k-1} \|\widetilde{\Phi}(k,j)\|
+ \frac{1}{k} \|S_v\| \sum_{j=1}^{k-1} \left( \alpha_j \,j\right)^2 + \frac{1}{k} e_i^\top (I-J) S_{\eta} (I-J)e_i \nonumber  \\
&=&
\label{eqn-sigma-k-new}
\frac{1}{N^2} {\bf SSNR}
+ 3 \|S_{\eta}\| \mathcal{Z}(k) + \|S_v\| \, \chi(k) +\frac{1}{k} e_i^\top (I-J) S_{\eta} (I-J)e_i.
\end{eqnarray}
Taking the $\limsup$ in \eqref{eqn-sigma-k-new} yields \eqref{eqn-sigma-2-k-2}.

We now prove \eqref{eqn-z-star}.
Note that $\mathcal{Z}(k)$
can be written via the following recursion:
  \begin{eqnarray*}
  \label{eqn-z-k-rec}
  \mathcal{Z}(k+1) &=& \left( 1 - \frac{b}{a+k+1}  \right) \left( \frac{k}{k+1} \mathcal{Z}(k) + \frac{1}{k+1} \right)\\
  \mathcal{Z}(1)   &=& 1 - \frac{b}{a+1}>0.
  \end{eqnarray*}
  The proof of~\eqref{eqn-z-star} proceeds analogously to the proof of \eqref{eqn-orthogonal},
  except that the vector quantity $(I-J)\epsilon(k)$ is replaced by the scalar $\mathcal{Z}(k)$, and
  the vector $m_{\eta}^{(1)}$ is replaced by the scalar $1$.
%
%
The proof of \eqref{eqn-chi-star} is trivial. Theorem~\ref{theorem-rate} now follows by combining~\eqref{eqn-mu-bound-lower} and~\eqref{eqn-sigma-2-k} with~\eqref{eqn-rate-p-e-i} to obtain~\eqref{eqn_mu_i_k_liminf}.
\end{proof}
%
\vspace{-3mm}
\section{Extensions to Non-Gaussian case}
\label{section-extensions}
We have characterized the exponential decay rate of the  error
probability in the case of Gaussian (spatially correlated and time-uncorrelated) sensing noise,
and Gaussian (spatially correlated and time-uncorrelated)
additive communication noise. Our results, to a certain degree, extend
 to the case when: 1) the zero mean sensing noise is
 spatially and temporally independent, but with a generic distribution
  with finite second moments; and 2)
  the zero mean additive communication noise is spatially correlated, temporally
  independent, and with a generic distribution and finite second moment.
  In this case, Theorem~\ref{theorem-rate}  remains valid,
  and all the steps in proving Theorem~\ref{theorem-rate}, equations (33)-(35)
  still go through in the generalized case also (see Appendix.) We now explain   the implications of Theorem~\ref{theorem-rate}  in the general non-Gaussian model.

Consider ${\bf{DSNR}}_{i}(k)$ in \eqref{eqn:DSNRisolated}. In the Gaussian case, ${\bf{DSNR}}_i(k)$
determines the exponential decay rate of the error probability,
as verified by equations \eqref{eqn-limsup} and \eqref{eqn-rate-p-e-i}.
In general, this is no longer the case as higher order moments play a role;
however, ${\bf{DSNR}}_i(k)$ still gives a good estimate for detection performance;
see,~e.g.,~\cite{scharf,Nick,k_factor_1}. With the optimal centralized
detector, we have that:
\begin{equation*}
\lim_{k \rightarrow \infty} \frac{{\bf{DSNR}}(k)}{k} = \frac{1}{4}{\bf SSNR}.
\end{equation*}
Theorem~\ref{theorem-rate}
implies that, with our distributed detector \eqref{eqn_recursive_algorithm}, the following holds for all sensors $i$:
\begin{equation}
\label{eqn-sada}
\liminf_{k \rightarrow \infty} \frac{{\bf{DSNR}_{i, \mathrm{dis}}}(k)}{k} \geq \frac{1}{4}{\bf SSNR}
\frac{\left( 1 -
\frac{\sqrt{N}}{b_0 \lambda_2(\mathcal{L})}c_{\mu} \right)^2}
{ \left( 1 + 3 \frac{N}{b_0 \lambda_2(\mathcal{L})} c_{\sigma} + \frac{N b_0^2}{\bf{G_c}}  \right)}.
\end{equation}
Eqn.~\eqref{eqn-sada} says that~${\bf{DSNR}}_{i,\mathrm{dis}}(k)$
 grows as $\Omega(k)$, as with the optimal centralized detector. This
contrasts with $\mathcal{MD}$ in \cite{Soummya-Detection-Noise} which achieves only $\Omega(k^\tau)$, $\tau<1.$ Second,
like before with Theorem \ref{theorem-rate}, now~\eqref{eqn-sada} reveals the tradeoff
between the information flow and communication~noise.
\vspace{-4mm}
\section{Conclusion}
\label{section-concl}
We designed a
consensus+innovations
distributed detector that achieves
exponential decay rate of the detection
 error probability \emph{at all sensors}
  under \emph{noisy} communication links,
   and even when certain (or most sensors)
    in isolation cannot perform successful detection. This
    improves over existing work like \cite{Soummya-Detection-Noise}
     that achieves a strictly slower rate.
     We showed how our distributed detector
     optimally weighs the neighbors'
      messages via the optimal sequence $\{\alpha_k\}$,
       balancing the two opposing effects: communication
       noise and information flow. We found
       a threshold on the communication noise power
       above which a sensor
       that successfully detects the event
        in isolation still improves its performance
        through cooperation over noisy links.

\vspace{-4mm}
\appendix
\section{Appendix}
\subsection{Proof of Theorem \ref{lemma-choice-beta}}
\label{append-beta}
Define first the following quantities:
\begin{eqnarray*}
\mathcal{Z}_{\beta}(k)  &:=&  \frac{1}{k} \sum_{j=1}^k \left( \Pi_{s=j}^k \left( 1 - \lambda_N(\mathcal{L}) \beta_s \right)^2 \right) ,\:\:\mathcal{Z}_{\beta}^\star := \lim_{k \rightarrow \infty} \mathcal{Z}_{\beta}(k)\\
{\chi}_{\beta}(k)  &:=&  \frac{1}{k} \sum_{j=1}^{k-1} \left( j\, \beta_j\right)^2,\:\:
{\chi}_{\beta}^\star  :=  \lim_{k \rightarrow \infty} {\chi}_{\beta}(k).
\end{eqnarray*}
We will need the following two Lemmas (\ref{lemma-new}
and \ref{lemma-yet-another-one}), of which
 Theorem \ref{lemma-choice-beta} is a direct~corollary.
\begin{lemma}
\label{lemma-new}
Let Assumptions \ref{assump:gaussnoises} through \ref{assump:alphak} hold. In addition, assume Assumption~\ref{assump:loc-det}. For the weight sequence $\beta_k = \frac{b_0}{a+k^\tau}$, we have the following:
 \begin{eqnarray}
 \label{eqn-z-beta}
 &\mathcal{Z}_{\beta}^\star&
 \left\{ \begin{array}{lll}
  =0  &\mbox{ if $\tau<1$} \\
   =1&\mbox{  if $\tau>1$} \\
   \geq \frac{1}{1 + 2 b_0 \lambda_N(\mathcal{L})}&\mbox{  if $\tau=1$}
       \end{array} \right.\\
       \label{eqn-chi-beta}
        &\chi_{\beta}^\star&
 \left\{ \begin{array}{lll}
  =+\infty  &\mbox{ if $\tau<1$} \\
  = 0&\mbox{  if $\tau>1$} \\
   =b_0^2&\mbox{  if $\tau=1.$}
       \end{array} \right.
 \end{eqnarray}
\end{lemma}
\begin{lemma}
\label{lemma-yet-another-one}
%
%
\vspace{-3.5 mm}
\begin{eqnarray}
\label{eqn-mu-i-k-NEW}
\lim_{k \rightarrow \infty} |\mu_i(k)| &=& \frac{1}{2}{\bf SSNR}_i\\
\label{eqn-sigma-2-k-NEW}
\liminf_{k \rightarrow \infty} k\,\sigma_i^2(k) &\geq&
\frac{1}{N} {\bf{SSNR}}_i +  {\bf{SSNR}}_i \mathcal{Z}_{\beta}^\star + \frac{\lambda_1(S_v)}{N^2} \chi_{\beta}^\star.
 \end{eqnarray}
\end{lemma}
\begin{proof}[Proof of Lemma \ref{lemma-new}]
We start by proving \eqref{eqn-z-beta} for $\tau<1$. To this end, note that
 $\mathcal{Z}_{\beta}(k)$ updates according to the following recursion:
 \begin{eqnarray}
 \label{eqn-tem-new}
 \mathcal{Z}_{\beta}(k+1) &=& \left( 1 - \frac{b^{\prime}}{a+(k+1)^\tau}  \right)^2
 \left( \frac{k}{k+1} \mathcal{Z}_{\beta}(k) + \frac{1}{k+1}\right)\\
 \mathcal{Z}_{\beta}(1) &=& \left(  1 - \frac{b^{\prime}}{a+1}\right)^2>0, \nonumber
 \end{eqnarray}
 where $b^\prime := b_0 \lambda_N(\mathcal{L})$.
By \eqref{eqn-tem-new}, for sufficiently large $k_0$, and for all $k \geq k_0$, we have that:
\begin{equation*}
\mathcal{Z}_{\beta}(k+1) \leq \left( 1 - \frac{b_{\beta}}{(k+1)^\tau} \right) \mathcal{Z}_{\beta}(k) + \frac{1}{k+1},
\end{equation*}
for appropriately chosen $b_{\beta}>0.$
Now, applying Lemma 4 in \cite{Soummya-Gossip}, we get that $\lim_{k \rightarrow \infty} \mathcal{Z}_{\beta}(k)=0.$

We proceed by proving \eqref{eqn-z-beta} for $\tau>1$. By \eqref{eqn-tem-new},
 and using the fact that $\mathcal{Z}_{\beta}(k) \leq 1$, $\forall k$, the quantity
 $\mathcal{Z}_{\beta}(k+1)$ can be bounded from below as follows:
 \begin{eqnarray}
 \label{eqn-koristimo}
 \mathcal{Z}_{\beta}(k+1) &\geq& \frac{k}{k+1} \left( 1 - \frac{2 b^\prime}{a+(k+1)^\tau}\right)
  \mathcal{Z}_{\beta}(k) + \left(  1 - \frac{2 b^\prime}{a+(k+1)^\tau} \right) \frac{1}{k+1}\\
  &=&
  \frac{k}{k+1} \mathcal{Z}_{\beta}(k) - \frac{2 b^\prime k}{(k+1)(a+(k+1)^\tau)} \mathcal{Z}_{\beta}(k)
 + \left(  1 - \frac{2 b^\prime}{a+(k+1)^\tau } \right)\frac{1}{k+1} \nonumber \\
 &\geq&
 \frac{k}{k+1} \mathcal{Z}_{\beta}(k) - \frac{2 b^{\prime \prime}}{k^{\tau}} + \frac{1}{k+1}, \nonumber
  \end{eqnarray}
 for appropriately chosen $b^{\prime \prime}>0$, and for all $k \geq k_1$, where
 $k_1$ is sufficiently large.
 Now, consider the recursion:
 \vspace{-3mm}
 \begin{eqnarray}
 \label{eqn-recursion-nova}
 \mathcal{U}(k+1) &=& \frac{k}{k+1} \mathcal{U}(k) + \frac{1}{k+1} - \frac{2 b^{\prime \prime}}{k^\tau},\:\:k=k_1,k_1+1,...\\
 \mathcal{U}(k_1) &=& \mathcal{Z}_{\beta}(k_1). \nonumber
 \end{eqnarray}
 Clearly, $\mathcal{Z}_{\beta}(k) \geq \mathcal{U}(k)$, for all $k \geq k_1$.
  Subtracting $1$ from both sides in \eqref{eqn-recursion-nova}
   and applying Lemma \ref{lemma-siegmund} yields $\mathcal{U}(k) \rightarrow 0$,
   and, hence, $\liminf_{k \rightarrow \infty} \mathcal{Z}_{\beta}(k) \geq 1$;
    on the other hand, $\mathcal{Z}_{\beta}(k) \leq 1$, for all $k$,
    and, hence, \eqref{eqn-z-beta} for $\tau>1$ holds.

    To prove \eqref{eqn-z-beta} and $\tau=1$, consider \eqref{eqn-koristimo}; as $\tau=1$, we have:
    \begin{equation*}
    \mathcal{Z}_{\beta}(k+1) \geq \frac{k}{k+1} \left( 1 - \frac{2 b^\prime}{a+k+1}\right)
    \mathcal{Z}_{\beta}(k) + \frac{1}{k+1} - \frac{2 b^\prime}{(k+1)^2},\:\:k=1,2,...
    \end{equation*}
    Now, define the recursion
    \begin{equation*}
    \mathcal{V}(k+1) = \frac{k}{k+1} \left( 1 - \frac{2 b^\prime}{a+k+1}\right)
    \mathcal{V}(k) + \frac{1}{k+1} - \frac{2 b^\prime}{(k+1)^2},\:\:k=1,2,...
    \end{equation*}
    Similarly to the proof of \eqref{eqn-prva-za-dokaz} in Theorem~\ref{theorem-rate},
    it can be shown that $\mathcal{V}(k) \rightarrow \frac{1}{1+2 b^\prime}$. Noting
    that $\mathcal{Z}_{\beta}(k) \geq \mathcal{V}(k)$, $k=1,2,...$ yields \eqref{eqn-z-beta} for $\tau=1$.


The proofs of \eqref{eqn-chi-beta} for $\tau<1$, $\tau>1$, and $\tau=1$ are trivial.
\end{proof}

\begin{proof}[Proof of Lemma \ref{lemma-yet-another-one}]
Note that, under the assumptions of Lemma \ref{lemma-yet-another-one},
 $S_{\eta} = {\bf{SSNR}}_i I$. Thus, in \eqref{eqn-sigma-temp},
  the term
  \begin{equation}
  \label{eqn-aux-eqn}
  \frac{2}{k^2}  \sum_{j=1}^{k-1} e_i^\top J S_{\eta} \widetilde{\Phi} (k,j) ^\top e_i =
  \frac{16}{k^2}   \frac{1}{8}{\bf{SSNR}}_i \sum_{j=1}^{k-1} e_i^\top J  \widetilde{\Phi} (k,j) ^\top e_i = 0,
  \end{equation}
because $J  \widetilde{\Phi}(k,j)=0$. Multiplying \eqref{eqn-sigma-temp} by $k$, and
using \eqref{eqn-aux-eqn}, we get:
%
\begin{equation}
\label{eqn-for-bounding}
k \sigma_i^2(k) = \frac{1}{N} {\bf{SSNR}}_i  +
\frac{{\bf{SSNR}}_i}{k}  \sum_{j=1}^k e_i^\top \widetilde{\Phi}(k,j) \widetilde{\Phi}(k,j)^\top e_i
+
\frac{1}{k} \sum_{j=1}^k \left( j\, \beta_j\right)^2 \left( e_i^\top {\Phi}(k,j) S_v {\Phi}(k,j)^\top e_i\right)
\end{equation}
We next bound $k \sigma_i^2(k) $ from above,
 using the following simple relations:
 \begin{eqnarray}
  b^\top A b &\geq & \lambda_1(A) \|b\|^2,\:\:
  \|e_i\| = 1 \nonumber \\
  \label{eqn-phi-temp}
  \| \Phi(k,j)^\top e_i \|^2 &\geq& \frac{1}{N^2},
   \end{eqnarray}
   where \eqref{eqn-phi-temp} holds true because $\Phi(k,j)$ is doubly stochastic. The upper bound
   on $k \sigma_i^2(k) $ is as follows:
\begin{eqnarray}
\label{eqn-for-bounding-2}
k \sigma_i^2(k) \hspace{-1.5mm } &\geq& \frac{1}{N} {\bf{SSNR}}_i  +
\frac{{\bf{SSNR}}_i}{k}  \sum_{j=1}^k \lambda_1 \left(  \widetilde{\Phi}(k,j) \widetilde{\Phi}(k,j)^\top \right)
+
\frac{1}{k} \sum_{j=1}^{k-1} \left( j\, \beta_j\right)^2 \lambda_1(S_v)
\|  {\Phi}(k,j+1)^\top e_i\|^2  \nonumber \\
&\geq&
\frac{1}{N} {\bf{SSNR}}_i  +
\frac{{\bf{SSNR}}_i}{k}  \sum_{j=1}^k \left( \Pi_{s=j}^k \left( 1 - \lambda_N(\mathcal{L}) \beta_s \right)^2 \right)
+
\frac{1}{k} \frac{1}{N^2} \sum_{j=1}^{k-1} \left( j\, \beta_j\right)^2 \lambda_1(S_v)
 \nonumber
\\
&=&
\frac{1}{N} {\bf{SSNR}}_i  +
\frac{{\bf{SSNR}}_i}{k}  \sum_{j=1}^k \left( \Pi_{s=k-j}^k \left( 1 - \lambda_N(\mathcal{L}) \beta_s \right)^2 \right)
+
\frac{1}{k} \frac{1}{N^2} \sum_{j=1}^{k-1} \left( j\, \beta_j\right)^2 \lambda_1(S_v)
 \nonumber\\
&=&
\label{eqn-j-na}
\frac{ {\bf{SSNR}}_i}{N}   +
{{\bf{SSNR}}_i} \,\mathcal{Z}_{\beta}(k)+
 \frac{\lambda_1(S_v)}{N^2} \,\mathcal{\chi}_{\beta}(k).
\end{eqnarray}
Taking the $\liminf$ in \eqref{eqn-j-na} yields \eqref{eqn-sigma-2-k-NEW}.
\end{proof}
\subsection{Decay rate of the  error probability for the $\mathcal{MD}$ algorithm in \cite{Soummya-Detection-Noise} }
\label{subsect-soummya}
We show that, under Gaussian assumptions, with the Algorithm $\mathcal{MD}$ in (\cite{Soummya-Detection-Noise},~eqn.~(14)),
 the  error probability decays to zero at a rate slower than exponential.
  Recall that $x_i(k)$, $\mu_i(k)$, $\sigma_i^2(k)$, and $P^e_i(k)$
    are the sensor $i$'s decision variable, its mean and variance under $H_1$,
    and its local error probability, respectively.
    (Now latter quantities correspond to $\mathcal{MD}$ and no longer to \eqref{eqn-running-cons-sensor-i}.) Denote by $P^e_w(k)$ the worst error probability at time $k$ among sensors:
\begin{equation}
\label{eqn=p-worst}
P^e_{w}(k) = \max_{i=1,...,N} P^e_i(k).
\end{equation}
We show that:\footnote{By Theorem \ref{theorem-rate}, with~\eqref{eqn-running-cons-sensor-i}, $P^e_{w}(k) = O(e^{-c k})$, hence better than $\mathcal{MD}.$ }
\vspace{-1mm}
\begin{equation}
\label{eqn-sub-exp}
P^e_{w}(k) = \Omega \left(e^{-c k^\tau}  \right),\:\:0.5<\tau<1,\:\:, c>0.
\end{equation}
Denote by $x(k)$ the vector of $x_i(k)$'s, as before, and $\Sigma(k):=\mathrm{Cov} \left( x(k)\right)$.
We will prove~\eqref{eqn-sub-exp}~by~showing:
\begin{equation}
\label{eqn-trace-order}
\mathrm{tr} \left( \Sigma(k) \right) = \Omega(1/k^\tau).
\end{equation}
Namely, with $\mathcal{MD}$,
$
\lim_{k \rightarrow \infty} \mu_i(k) = \frac{1}{2N} {\bf SSNR},\:\:\forall i.
$
Thus, for all $k \geq k^\prime$, for appropriate~$k^\prime>0$:
\begin{eqnarray}
P^e_{w}(k) &=& \mathcal{Q} \left(\min_{i=1,...,N} \frac{\mu_i(k)}{\sigma_i(k)}\right)
           \geq  \mathcal{Q} \left( \frac{\frac{1}{N}{\bf SSNR}}{\max_{i=1,...,N}\sigma_i(k)}\right)
           \geq  \mathcal{Q} \left(  \frac{\frac{1}{N}{\bf SSNR}}{\sqrt{\frac{1}{N}\mathrm{tr}(\Sigma(k))}}\right) \nonumber \\
           \label{eqn-log}
           &\geq&  \mathcal{Q} \left(  \frac{c_p}{\frac{1}{\sqrt{k^\tau}}}  \right),
\end{eqnarray}
for all $k \geq k^\prime$ and for appropriately chosen $c_p>0.$ Now,
applying the upper bound on the $\mathcal{Q}$ function in \eqref{eqn-tail-gauss} to \eqref{eqn-log}
yields~\eqref{eqn-sub-exp}. It remains to show \eqref{eqn-trace-order}. Denote by
 $\mathcal{W}(k): = I - \beta_k L - \alpha_k I$ the updating matrix in the
  $\mathcal{MD}$ algorithm, where $\beta_k=\frac{b}{(k+1)^\tau}$, $\tau \in (\frac{1}{2},1)$,
   and $\alpha_k = \frac{a}{(k+1)^\tau}$, $a,b>0$. In our notation,
    the update rule for $x(k)$ with $\mathcal{MD}$ is as follows:
    \begin{equation*}
    x(k+1) = \mathcal{W}(k) x(k) + \beta_k v(k) + \alpha_k \eta(k).
    \end{equation*}
   It can be shown that
  the covariance matrix $\Sigma(k):=\mathrm{Cov} \left( x(k) \right)$
   satisfies the following recurrent equation:
   \begin{equation}
   \label{eqn-trace}
   \Sigma(k+1) = \mathcal{W}(k) \Sigma(k) \mathcal{W}(k)^\top + \alpha_k^2 S_{\eta} +
   \beta_k^2 S_v.
   \end{equation}
   (Here $S_{\eta}$ and $S_v$
    denote respectively the covariance matrix of the innovations $\eta(k)$
     and of the communication noise $v(k)$, as before.) Taking the trace in \eqref{eqn-trace} and after algebraic manipulations, we get:
\begin{eqnarray}
\mathrm{tr} \left( \Sigma(k+1)\right) &\geq& \lambda_1(\mathcal{W}(k)^\top \mathcal{W}(k))
\, \mathrm{tr} \left( \Sigma(k) \right) + \alpha_k^2 \mathrm{tr}(S_{\eta}) + \beta_k^2 \mathrm{tr}(S_v) \nonumber\\
&\geq& \left( 1 - \alpha_k - \beta_k \lambda_N(\mathcal{L}) \right)^2 \mathrm{tr} \left( \Sigma(k) \right) +
 \beta_k^2 \mathrm{tr}(S_v) \nonumber \\
 &\geq& \left( 1 - 2 (\alpha_k+\beta_k \lambda_N(\mathcal{L})) \right) \mathrm{tr} \left( \Sigma(k) \right) +
 \beta_k^2 \mathrm{tr}(S_v) \nonumber \\
 &\geq& \left( 1 - \frac{c_{\Sigma}}{(k+1)^\tau} \right) \mathrm{tr}(\Sigma(k)) + \frac{c_v}{(k+1)^{2 \tau}} ,\nonumber
\end{eqnarray}
for all $k \geq k_2$ and $k_2$ sufficiently large,
for appropriately chosen $c_{\Sigma}, c_{v}>0$.
Now, introduce $\gamma(k):=\mathrm{tr}\left(\Sigma(k)\right)(k+1)^\tau.$
Then, we have:
\begin{equation*}
\gamma(k+1) \geq \left( 1 - \frac{c_{\Sigma}}{(k+1)^\tau} \right) \gamma(k) + \frac{c_v}{(k+1)^\tau}.
\end{equation*}
Now, consider the sequence $\mathcal{S}(k)$ that evolves according to the recursion:
\begin{equation}
\label{eqn-subtract}
\mathcal{S}(k+1) = \left( 1 - \frac{c_{\Sigma}}{(k+1)^\tau} \right) \mathcal{S}(k) + \frac{c_v}{(k+1)^{ \tau}},\:
k=k_2,k_2+1,..., \: \mathcal{S}(k_0)=\gamma(k_0).
\end{equation}
 Clearly, $\gamma(k) \geq \mathcal{S}(k) \geq 0$, for all $k=k_2,k_2+1,...$
It is easy to show that
\begin{equation}
\label{eqn-jna}
\lim_{k \rightarrow \infty} \mathcal{S}(k) = \frac{c_{v}}{c_{\Sigma}}.
\end{equation}
Namely, subtracting $\frac{c_{v}}{c_{\Sigma}}$ from both sides of equality
\eqref{eqn-subtract} yields:
\[
\left( \mathcal{S}(k+1) - \frac{c_{v}}{c_{\Sigma}} \right) = \left( 1 - \frac{c_{\Sigma}}{(k+1)^\tau} \right)
 \left( \mathcal{S}(k) - \frac{c_{v}}{c_{\Sigma}} \right) ,
\]
which in turn implies \eqref{eqn-jna}; now,
\eqref{eqn-jna} implies that $\gamma(k) = \Omega \left( 1\right)$. Hence, \eqref{eqn-trace-order} holds.
\vspace{-4mm}
\bibliographystyle{IEEEtran}
\bibliography{IEEEabrv,LDPbibliography2}
\end{document}